\documentclass[11pt]{article}
           
\usepackage{amsmath}
\usepackage{amsfonts}
\usepackage{amssymb}
\usepackage{mathtools}
\usepackage[mathscr]{eucal}
              \newif{\ifcomentarios}\comentariosfalse\newtheorem{theorem}{Theorem}

\newtheorem{definition}[theorem]{Definition}
\newtheorem{lemma}[theorem]{Lemma}
\newtheorem{remark}[theorem]{Remark}
\newtheorem{corollary}[theorem]{Corollary}
\newtheorem{proposition}[theorem]{Proposition}
\newenvironment{proof}[1][Proof]{\textit{#1.} }{\hfill $\Box$}

\newcommand{\OBSI}{\begin{remark}\begin{rm}}
\newcommand{\OBSF}{\end{rm}\end{remark}}
\newcommand{\DEFI}{\begin{definition}\begin{rm}}
\newcommand{\DEFF}{\end{rm}\end{definition}}
\newcommand{\zerarcounters}{\setcounter{equation}{0}\setcounter{theorem}{0}}

\newcommand{\dsum}{\displaystyle\sum}

\newcommand{\dint}{\displaystyle\int}

\newcommand{\be}{\begin{eqnarray}}
\newcommand{\en}{\end{eqnarray}}
\newcommand{\bee}{\begin{eqnarray*}}
\newcommand{\ene}{\end{eqnarray*}}

\DeclareMathOperator{\tr}{Tr}

\DeclareMathOperator{\img}{Rng}
\DeclareMathOperator{\rk}{rank}
\DeclareMathOperator{\supp}{supp}

\begin{document}
\title{Kotani Theory for ergodic matrix-like Jacobi operators}
\author{Fabr\'icio Vieira Oliveira \thanks{fabricio.vieira@engenharia.ufjf.br} \and Silas L. Carvalho\thanks{silas@mat.ufmg.br}}

\date{}\maketitle

{ \small \noindent $^{*}\,$  $^{\dag\,}$Instituto de Ciências Exatas (ICEX-UFMG). Av. Pres. Antônio Carlos 6627, Belo Horizonte-MG, 31270-901, Brasil.}

\begin{abstract}
We extend the so-called Kotani Theory for a particular class of ergodic matrix-like Jacobi operators defined in $l^{2}(\mathbb{Z}; \mathbb{C}^{l})$ by the law $[H_{\omega} \textbf{u}]_{n} := D^{*}(T^{n - 1}\omega) \textbf{u}_{n - 1} + D(T^{n}\omega) \textbf{u}_{n + 1} + V(T^{n}\omega) \textbf{u}_{n}$, where $T: \Omega \rightarrow \Omega$ is an ergodic automorphism in the measure space $(\Omega, \nu)$, the map $D: \Omega \rightarrow GL(l, \mathbb{R})$ is bounded, and for each $\omega\in\Omega$, $D(\omega)$ is symmetric. Namely, it is shown that for each $r\in\{1,\ldots,l\}$, the essential closure of $\mathcal{Z}_{r} := \{x \in \mathbb{R}\mid$ exactly  $2r$ Lyapunov exponents of $A_z$ are zero$\}$ coincides with $\sigma_{ac,2r}(H_{\omega})$, the  absolutely continuous spectrum of multiplicity $2r$, where $A_z$ is a Schr\"odinger-like cocycle induced by $H_\omega$. Moreover, if $k\in\{1,\ldots,2l\}$ is odd, then $\sigma_{ac,k}(H_{\omega})=\emptyset$ for $\nu$-a.e. $\omega\in\Omega$. We also provide a Thouless Formula for such class of operators.
\end{abstract}


\section{Introduction}
In the context of the spectral theory of discrete one-dimensional Schr\"odinger operators, defined in $l^{2}(\mathbb{Z}; \mathbb{C})$ by the action
\begin{equation}
\label{eq.schr.dis.esc}
(Hu)_{n} := u_{n + 1} + u_{n - 1} + v_{n}u_{n},
\end{equation}
where $(v_{n})_{n \in \mathbb{Z}}$ is a bilateral sequence of real numbers, there are some classical results on the characterization of their spectral components. 

One of these results refers to the case when $(v_{n})$ is dinamically defined by an ergodic automorphism, that is, given an invertible ergodic transformation $T: \Omega \rightarrow \Omega$ defined in a measure space $(\Omega, \nu)$ and a function $v: \Omega \rightarrow \mathbb{R}$, one sets, for each $\omega \in \Omega$, $v^{\omega}_{n} := v(T^{n}\omega)$; the resulting operator, $H_{\omega}$, is said to be ergodic. 
Kunz and Souillard have proven in this setting \cite{kunz1981} that there exist sets $\Sigma_{ac}, \Sigma_{sc}$ and $\Sigma_{pp}$ such that, for $\nu$-almost every $\omega \in \Omega$, 
\begin{equation*}
\begin{array}{lll}
\sigma_{ac}(H_{\omega}) & = & \Sigma_{ac}, \\
\sigma_{sc}(H_{\omega}) & = & \Sigma_{sc}, \\
\sigma_{pp}(H_{\omega}) & = & \Sigma_{pp},
\end{array}
\end{equation*}
where $\sigma_{ac}(H_{\omega}), \sigma_{sc}(H_{\omega})$ and $\sigma_{pp}(H_{\omega})$ are, respectively, the absolutely continuous, singular continuous and purely point spectral components of the operator $H_{\omega}$.  
Thus, one may say that the spectral components of the family of operators $(H_{\omega})_{\omega \in \Omega}$ are almost constant.

One may define, in connection with such class of operators, for each fixed $z \in \mathbb{C}$, the skew-linear product (cocycle) $\theta_{z}: \Omega \times \mathbb{C}^{2} \rightarrow \Omega \times \mathbb{C}^{2}$ by the law 
$\theta_{z} (\omega, \textbf{u}):=(T(\omega), A_{z}(\omega)\textbf{u})$, where $A_{z}: \Omega \rightarrow \mathrm{SL}(2, \mathbb{C})$ is given by 
\begin{equation}
\label{eq.cocli.schrodinger}
A_{z}(\omega) := 
\left[
\begin{array}{cc}
z - v(\omega) & -1 \\
1 & 0
\end{array}
\right].
\end{equation}

The orbits of the cocycle $\theta_{z}$ are associated with the solutions to the eigenvalue equation $(H_{\omega}u)_n=zu_n$. 
Namely, by considering the transfer matrices
\begin{equation}
\label{def.mat.trans.erg}
A_{n}(z, \omega) := 
\begin{cases}
A_{z}(T^{n -1}(\omega))A_{z}(T^{n - 2}(\omega))...A_{z}(T(\omega))A_{z}(\omega), & \mbox{ if }  n \geq 1, \\
& \\
\mathbb{I}_{2}, & \mbox{ if }  n = 0, \\
& \\
A^{-1}_{z}(T^{n}(\omega))...A^{-1}_{z}(T^{-2}(\omega))A_{z}^{-1}(T^{-1}(\omega)), & \mbox{ if }   n \leq - 1,
\end{cases}
\end{equation}
the sequence $(u_{n})_{n \in \mathbb{Z}}$ is a solution to the eigenvalue equation $(H_{\omega}u)_n=zu_n$ at $z\in\mathbb{C}$ iff
\begin{equation*}
\left[
\begin{array}{c}
u_{n + 1}\\
u_{n}
\end{array}
\right]
=
A_{n}(z, \omega)
\left[
\begin{array}{l}
u_{1}\\
u_{0}
\end{array}
\right].
\end{equation*}

One of the central results in the theory refers to the characterization of the absolutely continuous spectrum in terms of the Lyapunov exponent $\gamma(z)$ of the cocycle $\theta_{z}$, defined as
\begin{equation}
\label{eq.exp}
\gamma(z) := \lim_{n \rightarrow \infty} \frac{1}{n}\int\log \left\| A_{n}(z, \omega)\right\|d\nu(\omega)= \lim_{n \rightarrow \infty} \frac{1}{n}\log \left\| A_{n}(z, \omega)\right\|,
\end{equation}
with the second identity valid for $\nu$-a.e $\omega\in\Omega$ (as a consequence of the ergodicity of $(\Omega,T,\nu)$), obtained by Kotani in \cite{kotani82} for continuous operators, and then by Simon in~\cite{simon83} for the family $H_\omega$ defined as above. The result establishes that, for $\nu$-a.e. $\omega\in\Omega$,
\[\Sigma_{ac}=\overline{\{z\in\mathbb{R}\mid\gamma(z)=0\}}^{ess},
\]
where $\overline{S}^{ess}:=\{E\in\mathbb{R}\mid\kappa(S\cap(E-\varepsilon,E+\varepsilon))>0$, for every $\varepsilon>0\}$ is the essential closure of $S\subset\mathbb{R}$ with respect to the Lebesgue measure $\kappa$. 

This result was extended by Kotani and Simon \cite{kotani88} to matrix-valued Schr\"odinger operators defined in $l^{2}(\mathbb{Z}; \mathbb{C}^{l})$, for each $\omega \in \Omega$, by the law
\begin{equation} 
\label{eq.ope.din.schr}
[H^S_{\omega} \textbf{u}]_{n} := \textbf{u}_{n - 1} + \textbf{u}_{n + 1} + V(T^{n}\omega) \textbf{u}_{n},
\end{equation}
where $V: \Omega \rightarrow M(l, \mathbb{R})$ is a map with range in the set of self-adjoint matrices and, again,  $T: \Omega \rightarrow \Omega$ is an ergodic automorphism in the measure space $(\Omega, \nu)$.

Our main goal in this work is to go a step further, extending such characterization of the absolutely continuous spectrum (including multiplicity) to a family of operators of the form
\begin{equation} 
\label{eq.ope.din.jacobi}
[H_{\omega} \textbf{u}]_{n} := D^{*}(T^{n - 1}\omega) \textbf{u}_{n - 1} + D(T^{n}\omega) \textbf{u}_{n + 1} + V(T^{n}\omega) \textbf{u}_{n}, 
\end{equation}
where $D: \Omega \rightarrow GL(l, \mathbb{R})$ is a symmetric and invertible $l\times l$ matrix with real entries. This is a particular case of the so-called matrix-valued Jacobi operators, defined in \cite{marx15}. Such operators arise naturally in the study of some quasi-crystals which satisfy the Aubry-Andr\'e duality (see~\cite{marx15} for details and main motivations), and can be seen as the matrix-valued version of the one-dimensional Jacobi operators. 

For each $z \in \mathbb{C}$ and each $\omega \in \Omega$, if one considers the eigenvalue equation 
\begin{equation}
\label{eq.autovalor}
D^{*}(T^{n - 1}\omega) \textbf{u}_{n - 1} + D(T^{n}\omega) \textbf{u}_{n + 1} + V(T^{n}\omega) \textbf{u}_{n} = z \textbf{u}_{n}
\end{equation}
associated with the operator $H_{\omega}$, one may (like in \cite{marx15}) write this equation in the form 
\begin{equation*}
\left[
\begin{array}{c}
\textbf{u}_{n + 1} \\
D^{*}(T^{n}\omega)\textbf{u}_{n}
\end{array}
\right]
=
A_{z}(T^{n}\omega)
\left[
\begin{array}{c}
\textbf{u}_{n} \\
D^{*}(T^{n - 1}\omega)\textbf{u}_{n - 1}
\end{array}
\right],
\end{equation*}
with $A_{z}: \Omega \rightarrow GL(2l, \mathbb{C})$ (as in \eqref{eq.cocli.schrodinger}) given by the law
\begin{equation}
\label{eq.cociclo.az}
A_{z}(\omega) := 
\left[
\begin{array}{cc}
D^{-1}(\omega)(z - V(\omega)) & -D^{-1}(\omega) \\
&\\
D^{*}(\omega) & 0
\end{array}
\right];
\end{equation}
note that for each $\omega\in\Omega$, $A_z(\omega)$ is a symplectic complex matrix of size $2l \times 2l$, that is, $(A_z(\omega))^t\mathbb{J}A_z(\omega)=\mathbb{J}$, where
\[\mathbb{J}:=\left[
\begin{array}{cc}
0 & \mathbb{I} \\
- \mathbb{I} & 0
\end{array}
\right].\]
One may define the transfer matrices $A_{n}(z, \omega)$ as in \eqref{def.mat.trans.erg}.

If one also assumes that for each $z\in\mathbb{C}$ the cocycle $(T, A_{z})$ is such that $\log^{+} \left\| A_z (\cdot)\right\| \in L^{1}(\nu)$, then it follows from Oseledec Theorem~\cite{oseledets68} (see also~\cite{ruelle79}) that for each $j=1,\ldots,2l$, the so-called $j$-th Lyapunov exponent of the cocycle 
is well defined by the law 
\begin{equation}
\label{eq.ruelle}
\gamma_{j}(z, \omega)=\lim_{n \rightarrow \infty} \frac{1}{n} \log \left(s_{j}[A_{n}(z, \omega)]\right),
\end{equation} 
where $s_j(A_n)$ stands for the $j$-th singular value of $A_n$ (they are ordered so that $s_j(A_n)\ge s_{j+1}(A_n)$), if the limit exists. The ergodicity of $T$ guarantees that, for $z \in \mathbb{C}$ fixed, $\gamma_{j}(z, \omega)$ is constant for $\nu$-a.e. $\omega\in\Omega$. 

Namely, we prove, like in~\cite{kotani88}, the following result.

\begin{theorem}\label{maintheo}
  Let $(H_\omega)_\omega$ be a family of ergodic matrix-valued Jacobi operators of the form~\eqref{eq.ope.din.jacobi} such that the map $D: \Omega \rightarrow GL(l, \mathbb{R})$ is bounded and for each $\omega\in\Omega$, $D(\omega)$ is a symmetric and invertible $l\times l$ matrix. Suppose also that the map $A_z$, given by the law~\eqref{eq.cociclo.az}, is such that $\log^{+} \left\| A_z (\omega)\right\|\in L^{1}(\nu)$. Then, for each $r\in\{1,\ldots,l\}$, the essential closure of
\begin{equation}
\label{def.sup.exp}
\mathcal{Z}_{r} := \{x \in \mathbb{R}\mid \mbox{ exactly  } 2r \mbox{ exponents }  \gamma_{j}(x) \mbox{ are zero} \} 
\end{equation}
coincides with $\sigma_{ac,2r}(H_{\omega})$, the  absolutely continuous spectrum of multiplicity $2r$. Moreover, if $k\in\{1,\ldots,2l\}$ is odd, then $\sigma_{ac,k}(H_{\omega})=\emptyset$ for $\nu$-a.e. $\omega\in\Omega$.
\end{theorem}

The inclusion $\overline{\mathcal{Z}_{r}}^{ess} \subseteq \sigma_{ac,r}(H_{\omega})$ is known in the literature as Kotani Theorem (this is Theorem \ref{teo.ko}); the inclusion $\sigma_{ac,r}(H_{\omega}) \subseteq \overline{\mathcal{Z}_{r}}^{ess}$ is the so-called Ishii-Pastur Theorem  (Theorem \ref{teo.pastur}).

\begin{remark}
We note that the additional hypotheses required for the map $D$ are needed  to guarantee some desirable properties for the operator $H_{\omega}$. Specifically, for each $\omega \in \Omega$, the hypothesis of $D(\omega)$ being invertible  is necessary in \eqref{eq.cociclo.az}, the hypothesis of $D(\omega)$ being symmetric is sufficient, for instance, to establish Green Formula \eqref{eq.green.1}, and the hypothesis of $D$ being a bounded map implies \eqref{criterio.extensao}, a sufficient condition for the operator $H_\omega$ to be in the limit point case.
\end{remark}  

In order to prove Kotani Theorem, one must show that the orthogonal derivative, with respect to $z$, of the sum of the Lyapounov exponents (at $+\infty$) up to $\gamma_j(x)$, for each $j\in\{1,\ldots,l\}$, exists and is finite for $\kappa$-a.e. $x\in\mathbb{R}$. For that, one might use the strategy presented in the proof of the so-called Thouless Formula (we prove the following version of this formula in Theorem \ref{teo.thou}): 
\begin{equation}
\label{eq.form.thou}
\sum^{l}_{j = 1}\gamma_{j}(z) = \int_{\mathbb{R}} \log \left|z - x\right| dk(x) - \int_{\Omega} \log \left\vert\det \left( D(\omega) \right)\right\vert d\nu(\omega),
\end{equation}
where $k$ is the integrated density of states (see Definition \ref{def.dens.est}) (note that $ \log \left\vert\det \left( D(\cdot) \right)\right\vert\in L^1(\nu)$, since $D$ is bounded and for each $\omega\in\Omega$, $D(\omega)$ is invertible). Then, given that the sum of the Lyapunov exponents is, as a function of $z$, the Borel transform of a measure (see the proof of Corollary~\ref{coro.thouless}), it has the desired regularity.

Thouless Formula was first established in \cite{thou1972} for one-dimensional Schr\"odinger operators. We note that there is a proof of Thouless Formula for operators of the form \eqref{eq.ope.din.schr} in \cite{kotani88}; we have used this proof as a guide for the proof of our result. We also note that there is, in \cite{haro13}, an alternative proof of the formula for operators of the form \eqref{eq.ope.din.jacobi} when the map $D: \Omega \rightarrow M(l, \mathbb{R})$ is constant.

The organization of this paper is as follows. In Section~\ref{resolvente} we obtain, among other results, Green Formula for $(H_\omega)_\omega$, we present a sufficient condition for $(H_\omega)_\omega$ to be in the limit-point case, and then we present a characterization of the absolutely continuous spectrum of multiplicity $j\in\{1,\ldots,l\}$ of $H_\omega$ with respect to the equivalent spectra for $H_{\omega}^{\phi}$ and $H_{\omega}^{\phi,-}$, the restrictions of $H_\omega$ to $l^2(\mathbb{Z}_+;\mathbb{C}^l)$ and $l^2(\mathbb{Z}_+;\mathbb{C}^l)$, respectively, satisfying  Dirichlet boundary condition at $n=0$.

In Section~\ref{thouless} we represent the Lyapounov exponents in terms of the so-called Jost solutions and then prove Thouless Formula. In Section~\ref{kotani} we present several auxiliary results that are used in the proof of Kotani Theorem, which we also prove there. Finally, in Section~\ref{pastur} we prove Ishii-Pastur Theorem, and then complete the proof of Theorem~\ref{maintheo}.

\section{The Resolvent Operator}
\label{resolvente}
\zerarcounters

\subsection{Green Formula and self-adjoint extensions of $(H_\omega)_\omega$}

Let $(H_\omega)_\omega$ be a family of dynamically defined operators of the form  \eqref{eq.ope.din.jacobi} such that $T:\Omega\rightarrow\Omega$ is an ergodic automorphism and, for each $\omega\in\Omega$ and each $n\in \mathbb{Z}$, $V^{\omega}_{n} := V(T^{n}\omega)$ and $D^{\omega}_{n} := D(T^{n} \omega)$ are invertible and symmetric real matrices. Then, for each $\textbf{u}, \textbf{v} \in (\mathbb{C}^{l})^{\mathbb{Z}}$, each $n,m\in\mathbb{Z}$ such that $n>m$, and each $\omega\in\Omega$, one has the so-called Green Formula for $H_\omega$:
\begin{equation}
\label{eq.green.1}
\sum^{n}_{k = m} \left\langle (H_{\omega}\textbf{u})_{k}, \bar{\textbf{v}}_{k} \right\rangle_{\mathbb{C}^{l}} - \left\langle (H_{\omega}\textbf{v})_{k}, \bar{\textbf{u}}_{k} \right\rangle_{\mathbb{C}^{l}} = W^{\omega}_{[\textbf{u}, \textbf{v}]}(n + 1) - W^{\omega}_{[\textbf{u}, \textbf{v}]}(m),
\end{equation}
where 
\begin{equation*}
W^{\omega}_{[\textbf{u}, \textbf{v}]}(n) :=  \left\langle D^{\omega}_{n - 1}\textbf{u}_{n}, \bar{\textbf{v}}_{n - 1} \right\rangle_{\mathbb{C}^{l}} - \left\langle D^{\omega}_{n - 1} \textbf{v}_{n}, \bar{\textbf{u}}_{n - 1}  \right\rangle_{\mathbb{C}^{l}}
\end{equation*}
is the so-called Wronskian of $\textbf{u}$ and $\textbf{v}$ at $n\in\mathbb{Z}$.

When one represents $\textbf{u}_{n}, \textbf{v}_{n}$ as column-vectors, the Wronskian of $\textbf{u}$ and $\textbf{v}$ at $n\in\mathbb{Z}$ may be written as
\begin{equation*}
W^{\omega}_{[\textbf{u}, \textbf{v}]}(n) =  \textbf{u}^{t}_{n}D^{\omega}_{n - 1}\textbf{v}_{n - 1} - \textbf{v}^{t}_{n}D^{\omega}_{n - 1}\textbf{u}_{n - 1}.
\end{equation*}

Finally, if $(A_{n}), (B_{n})$ are sequences of matrices of size $l \times l$ one obtains, by applying $H_\omega$ to each of their columns, a generalization of Green Formula:
\begin{equation}
\label{wronski.matriz}
\sum^{n}_{k = m}(A^{t}_{k} H_{\omega}(B)_{k}  - H_{\omega}(A)_{k}^{t} B_{k})  =  W^{\omega}_{[A, B]}(n + 1) - W^{\omega}_{[A, B]}(m),
\end{equation}
with
$$
W^{\omega}_{[A, B]}(m) := A^{t}_{m - 1}D^{\omega}_{m - 1}B_{m} - A^{t}_{m}D^{\omega}_{m - 1}B_{m - 1}.
$$


\begin{lemma}[Constancy of the Wronskian]\label{const.wronsk}
  If $\textbf{u}, \textbf{v} \in (\mathbb{C}^{l})^{\mathbb{Z}}$ are solutions to the eigenvalue equation~\eqref{eq.autovalor} at 
  $z \in \mathbb{C}$, then the Wronskian $W_{[\textbf{u}, \textbf{v}]}(n)$ is constant.
\end{lemma}
\begin{proof}
This result is a direct consequence of Green Formula. Namely, for each integers $n>m$, one has
\begin{eqnarray*}
\sum^{n}_{k = m} \left\langle (H\textbf{u})_{k}, \bar{\textbf{v}} \right\rangle_{\mathbb{C}^{l}} - \left\langle (H\textbf{v})_{k},\bar{\textbf{u}}_{k} \right\rangle_{\mathbb{C}^{l}}  =  \sum^{n}_{k = m} \left\langle z \textbf{u}_{k},\bar{\textbf{v}}_{k} \right\rangle_{\mathbb{C}^{l}} - \left\langle z \textbf{v}_{k}, \bar{\textbf{u}}_{k} \right\rangle_{\mathbb{C}^{l}} = 0,
\end{eqnarray*}
from which follows that
\[
W_{[\textbf{u}, \textbf{v}]}(n + 1) - W_{[\textbf{u}, \textbf{v}]}(m) = 0. 
\]
\end{proof}

\begin{lemma}
\label{lema.wrons.zx}
Let $z=x+iy\in\mathbb{C}$. If $\textbf{u}, \textbf{v}, \textbf{r}, \textbf{s}  \in (\mathbb{C}^{l})^{\mathbb{Z}}$ are such that $H \textbf{u} = z \textbf{u}, H \textbf{r} = z \textbf{r}$, $H \textbf{v} = x \textbf{v}$ and $H\textbf{s} = x \textbf{s}$, then
\[W_{[\textbf{u} - \textbf{v}, \bar{\textbf{r}} - \bar{\textbf{s}}]}(n + 1) - W_{[\textbf{u} - \textbf{v}, \bar{\textbf{r}} - \bar{\textbf{s}}]}(m) = i y \sum^{n}_{k = m}\left( 2 \left\langle  \textbf{u}_{k}, \textbf{r}_{k} \right\rangle_{\mathbb{C}^{l}} - \left\langle  \textbf{u}_{k}, \textbf{s}_{k} \right\rangle_{\mathbb{C}^{l}} - \left\langle  \textbf{v}_{k} , \textbf{r}_{k} \right\rangle_{\mathbb{C}^{l}}\right).\]
\end{lemma}
\begin{proof}
One has, for each $k \in \mathbb{Z}$, the identity
\begin{eqnarray*}
\begin{array}{lll}
& \left\langle (H(\textbf{u} - \textbf{v}))_{k}, (\textbf{r} - \textbf{s})_{k} \right\rangle_{\mathbb{C}^{l}} - \left\langle (H(\bar{\textbf{r}} - \bar{\textbf{s}}))_{k}, (\bar{\textbf{u}} - \bar{\textbf{v}})_{k} \right\rangle_{\mathbb{C}^{l}} = & \\
& & \\
 &  \left\langle z \textbf{u}_{k} - x \textbf{v}_{k} , \textbf{r}_{k} - \textbf{s}_{k} \right\rangle_{\mathbb{C}^{l}} - \left\langle \bar{z} \bar{\textbf{r}}_{k} - x \bar{\textbf{s}}_{k} , \bar{\textbf{u}}_{k} - \bar{\textbf{v}}_{k} \right\rangle_{\mathbb{C}^{l}} = & \\
& & \\
 &  \left\langle z \textbf{u}_{k}, \textbf{r}_{k} \right\rangle_{\mathbb{C}^{l}} - \left\langle z \textbf{u}_{k} , \textbf{s}_{k} \right\rangle_{\mathbb{C}^{l}}  +  \left\langle x \textbf{v}_{k} , \textbf{s}_{k} \right\rangle_{\mathbb{C}^{l}} - \left\langle x \textbf{v}_{k} , \textbf{r}_{k}  \right\rangle_{\mathbb{C}^{l}} + &  \\
&&\\
&  \left\langle \bar{z} \bar{\textbf{r}}_{k}, \bar{\textbf{v}}_{k} \right\rangle_{\mathbb{C}^{l}} - \left\langle \bar{z} \bar{\textbf{r}}_{k}, \bar{\textbf{u}}_{k} \right\rangle_{\mathbb{C}^{l}}  +  \left\langle x \bar{\textbf{s}}_{k} , \bar{\textbf{u}}_{k} \right\rangle_{\mathbb{C}^{l}} - \left\langle x \bar{\textbf{s}}_{k} , \bar{\textbf{v}}_{k} \right\rangle_{\mathbb{C}^{l}}. &
\end{array}
\end{eqnarray*}

Now, by summing both members in the previous identity from $m$ up to $n$, it follows that
\begin{eqnarray*}
\begin{array}{lll}
 & \sum^{n}_{k = m} \left\langle (H(\textbf{u} - \textbf{v}))_{k}, (\textbf{r} - \textbf{s})_{k} \right\rangle_{\mathbb{C}^{l}} - \left\langle (H(\bar{\textbf{r}} - \bar{\textbf{s}}))_{k}, (\bar{\textbf{u}} - \bar{\textbf{v}})_{k} \right\rangle_{\mathbb{C}^{l}} = & \\
&&\\
& \sum^{n}_{k = m} 2yi \left\langle \textbf{u}_{k} , \textbf{r}_{k} \right\rangle_{\mathbb{C}^{l}}  - \sum^{n}_{k = m} yi \left\langle  \textbf{u}_{k}, \textbf{s}_{k} \right\rangle_{\mathbb{C}^{l}} - \sum^{n}_{k = m} yi \left\langle  \textbf{v}_{k} , \textbf{r}_{k} \right\rangle_{\mathbb{C}^{l}}.
\end{array}
\end{eqnarray*}

The result now is a consequence of Green Formula.
\end{proof}

\

In what follows, it will be necessary to guarantee that the operator $H_\omega$ is in the limit point case at $+\infty$; we proceed as in~\cite{kotani88}.

The so-called boundary form of the restriction of the operator $H_{\omega}$ to $\mathbb{Z}_{+}$ is defined as  
\begin{equation*}
\begin{array}{lll}
\Gamma_{\omega}(\textbf{u}, \textbf{v}) & := & \displaystyle\sum^{\infty}_{k = 1} \left\langle (H_{\omega}\textbf{u})_{k}, \bar{\textbf{v}}_{k} \right\rangle_{\mathbb{C}^{l}} - \left\langle (H_{\omega}\textbf{v})_{k}, \bar{\textbf{u}}_{k} \right\rangle_{\mathbb{C}^{l}} =\\
& & \\
&   & \displaystyle\lim_{n \rightarrow \infty} \left( \left\langle \textbf{v}_{n}, D^{\omega}_{n}\textbf{u}_{n + 1} \right\rangle_{\mathbb{C}^{l}} -  \left\langle \textbf{v}_{n + 1}, D^{\omega}_{n}\textbf{u}_{n} \right\rangle_{\mathbb{C}^{l}} \right) -\\
&&\\
&& \left(  \left\langle \textbf{v}_{0}, D^{\omega}_{0}\textbf{u}_{1} \right\rangle_{\mathbb{C}^{l}} - \left\langle \textbf{v}_{1}, D^{\omega}_{0}\textbf{u}_{0} \right\rangle_{\mathbb{C}^{l}} \right).
\end{array}
\end{equation*} 

Since $D$ is a bounded map, it follows that for every $\textbf{u}, \textbf{v} \in l^{2}(\mathbb{N};\mathbb{C}^{l})$,  
\begin{equation}
\label{criterio.ponto}
\lim_{n \rightarrow \infty} \left( \left\langle \textbf{v}_{n}, D^{\omega}_{n}\textbf{u}_{n + 1} \right\rangle_{\mathbb{C}^{l}} -  \left\langle \textbf{v}_{n + 1}, D^{\omega}_{n}\textbf{u}_{n} \right\rangle_{\mathbb{C}^{l}} \right) = 0.
\end{equation}

This conditition implies that the deficiency indices of the restriction of the operator to $\mathbb{Z}_{+}$ are equal to $l$, a property known in the literature as the limit point case (see~\cite{carmona90,cesar2009,teschl00}).

In this case, the self-adjoint extensions of $H_{\omega}$ are associated with the domains $\mathcal{D}$ for which the boundary form is trivial (see Proposition~7.1.3 in~\cite{cesar2009}), that is, for which
\begin{equation}
\label{criterio.extensao}
\left\langle \textbf{v}_{0}, D^{\omega}_{0}\textbf{u}_{1} \right\rangle_{\mathbb{C}^{l}} - \left\langle \textbf{v}_{1}, D^{\omega}_{0}\textbf{u}_{0} \right\rangle_{\mathbb{C}^{l}} = 0,
\end{equation}
for every $\textbf{u}, \textbf{v} \in \mathcal{D}$. We are particularly interested in the extension that satisfies the initial condition $\textbf{u}_{0} = 0$, which is denoted by $H^{\phi}_{\omega}$ and called \textit{Dirichlet operator}. We also consider the extension with the initial condition $\textbf{u}_{1} = 0$, denoted by $H^{\psi}_{\omega}$ and called \textit{Neumann operator}.

Associated with Dirichlet and Neumann operators, one defines sequences of matricial solutions to the eigenvalue equation~\eqref{eq.autovalor}, $\psi(z, \omega), \phi(z, \omega)$, called respectively Neumann and Dirichlet solutions; these are sequences of matrices of size $l \times l$ such that 
\begin{equation}
\label{def.neu.dir}
\begin{cases}
\psi_{0}(z, \omega) = \mathbb{I}, &  \\
\psi_{1}(z, \omega) = 0, &  
\end{cases}
\begin{array}{lll}
 &  &  \\
 &  &   
\end{array}
\begin{cases}
\phi_{0}(z, \omega) = 0,& \\
\phi_{1}(z, \omega) = \mathbb{I},& 
\end{cases}
\end{equation}
and whose columns satisfy \eqref{eq.autovalor}. 


\subsection{Integral Kernel of the Resolvent Operator}

As usual, one may obtain the spectral properties of the operator $H^{\phi}_{\omega}$ by analysing the asymptotic behavior of the resolvent operator for $z = x + iy \in \mathbb{C}$ as $y \downarrow 0$.

The idea is to 
write, for each $z\in\rho(H^{\phi}_{\omega})$ (the resolvent set of $H^{\phi}_{\omega}$), the resolvent operator in its integral form through the matrix-valued Green Function, which is, by its turn, parametrized by the solutions to the eigenvalue equation at $z$. 

If $z\in \mathbb{C}\setminus \mathbb{R}$ and if $H^\phi_\omega$ is in the limit point case \eqref{criterio.ponto}, then the set 
\[
\mathcal{J}_{+}(z, \omega):=\{\textbf{u} \in (\mathbb{C}^{l})^{\mathbb{Z}} \; \vert \; H_{\omega}\textbf{u} = z\textbf{u}, \;\sum^{\infty}_{n = 1} \left\|\textbf{u}_{n}\right\|^{2} < \infty\}
\]
is a subspace of dimension $l$. In this case, there exists only one sequence $(F_{n}(z, \omega))_{n}$ of matrices of size $l \times l$, whose columns satisfy simultaneously \eqref{eq.autovalor}, the initial condition $F_{0}(z, \omega)$ = $\mathbb{I}$, and
\begin{equation}
\label{def.jost}
\sum^{\infty}_{n = 0} \left\|F_{n}(z, \omega) \right\|^{2} < \infty.
\end{equation}

Each sequence of columns of $F_{n}(z, \omega)$ is a solution to the eigenvalue equation with a canonical vector of $\mathbb{C}^l$ as its initial condition; namely, if $F_n^j(z,\omega)$ denotes the $j$-th column of $F_n(z,\omega)$, then $(F_{n}^j(z, \omega))_n$, with $F_0^j(z,\omega)=e^j$ (where $e^j$ stands for the $j$-th element of the canonical basis of $\mathbb{C}^l$), is a solution to the eigenvalue equation~\eqref{eq.autovalor}. These $l$ solutions are the so-called \textit{Jost solutions}.   

The Jost solutions can be parametrized by the matrix-valued Weyl-Titchmarsh function associated with $H^{\phi}_\omega$, $M^{\phi}(z,\omega)$, which is given by
\begin{equation}
\label{def.m.weyl}
M^{\phi}(z, \omega) := - F_{1}(z, \omega)(D^{\omega}_{0})^{-1}.
\end{equation}

One can also write the Jost solutions in terms of Neumann and Dirichlet solutions as
\begin{equation}
\label{equa.m.jost}
F_{n}(z, \omega) = \psi_{n}(z, \omega) - \phi_{n}(z, \omega) M^{\phi}(z, \omega) D^{\omega}_{0};
\end{equation}
this is a consequence of the fact that both members of \eqref{equa.m.jost} are solutions to the eigenvalue equation that coincide at $n\in\{0,1\}$ 
(given the unicity of solutions to such equations).

In what follows, we establish a relation between the matrix-valued Weyl-Titchmarsh function and the Jost solutions, as it was done in \cite{kotani88} (see Proposition $2.3$ there).

\begin{proposition}
\label{prop.m}
Let, for each $z \in \mathbb{C} \setminus \mathbb{R}$ and each $\omega \in \Omega$, $\left(F_{n}(z,\omega)\right)_{n}$ be the matrices given by the Jost solutions to the eigenvalue equation~\eqref{eq.autovalor},  and let $M^{\phi}(z,\omega)$ be the corresponding matrix-valued Weyl-Titchmarsh function. Then,
\begin{eqnarray*}
  (a) && (M^{\phi}(z,\omega))^{t} = M^{\phi}(z,\omega);\\
  (b) && D_{0}\Im[M^{\phi}(z,\omega)]D_{0} = \Im[z] \sum^{\infty}_{k = 1} F_{k}(z,\omega)^{*}F_{k}(z,\omega).
\end{eqnarray*}
\end{proposition}

\begin{proof} We omit $z$ and $\omega$ in the notation throughout the proof.

$(a)$ By letting $A_{n} = B_{n} = F_{n}$ in Green Formula \eqref{wronski.matriz}, it follows that there exists a constant $C$ such that, for every $n \in \mathbb{Z}_{+}$,
\[
W_{[F, F]}(n) = (F_{n - 1}^{t}D_{n - 1}F_{n} - F_{n}^{t}D_{n - 1}F_{n - 1}) = C;
\]
one has from relation \eqref{def.jost} that
\[
\lim_{n \rightarrow \infty} W_{[F, F]}(n) = 0,
\]
and then, $C = 0$. It follows, in particular, that  
\[
(D_{0}F_{0})^{t}F_{1} - (D_{0}F_{1})^{t}F_{0} = 0.
\]
Since $F_{0} = \mathbb{I}$ and $D_{0} = D_{0}^{t}$, it follows that $D_{0}F_{1} = D_{0}M^{\phi}D_{0}$ is symmetric, and so $M^{\phi}$.

$(b)$ Let $A_{n} = \overline{F}_{n}$, $B_{n} = F_{n}$ and $m = 1$ in~\eqref{wronski.matriz}; then, 
take the limit $n \rightarrow \infty$. It follows from~\eqref{criterio.ponto},~\eqref{def.m.weyl} and item $(a)$ that 
\[
D_{0}\Im[M^{\phi}(z,\omega)]D_{0}=(D_{0}\overline{F}_{0})^{t}F_{1} - (D_{0}\overline{F}_{1})^{t}F_{0}  = \Im[z] \sum^{\infty}_{k = 1}  F_{k}^{*}F_{k}.
\]
\end{proof}

\begin{lemma}
\label{lema.green.nd}
For each $z \in \mathbb{C}$ and each $\omega \in \Omega$, let  $\psi(z,\omega)$ and $\phi(z,\omega)$ be the Neumann and Dirichlet solutions to the eigenvalue equation~\eqref{eq.autovalor} at $z$, respectively. Then, for each $n \in \mathbb{Z}_{+}$, one has
\begin{eqnarray}\label{eq.F.phi}
  \begin{array}{ll}
(a) & \psi_{n}(z,\omega)(D^\omega_0)^{-1}\phi_n^t(z,\omega) - \phi_{n}(z,\omega)(D^\omega_0)^{-1}\psi_{n}^t(z,\omega)  =  0;\\

(b) & \psi_{n}(z,\omega)(D^\omega_0)^{-1}\phi_{n + 1}^t(z,\omega) - \phi_{n}(z,\omega) (D^\omega_0)^{-1} \psi_{n + 1}^t(z,\omega)  =  (D^\omega_n)^{-1};\\ 

(c) & \psi_{n + 1}(z,\omega)(D^\omega_0)^{-1}\phi_{n}^t(z,\omega) - \phi_{n + 1}(z,\omega)(D_{0}^\omega)^{-1}\psi_{n}^t(z,\omega)  =  - (D^\omega_n)^{-1}. 
\end{array}\end{eqnarray}
\end{lemma}

\begin{proof} Again, we omit $z$ and $\omega$ in the notation throughout the proof. If one applies equation \eqref{wronski.matriz} to the pairs $(\psi, \psi), (\psi, \phi), (\phi, \psi)$ and $(\phi, \phi)$, one obtains from the constancy of the Wronskian (Lemma~\ref{const.wronsk}), for each $n \in \mathbb{Z}_{+}$, the system
\begin{equation*}
\begin{cases}
\psi^{t}_{n}D_{n}\psi_{n + 1} - \psi^{t}_{n + 1}D_{n}\psi_{n} = 0, & \\ 
\psi^{t}_{n}D_{n}\phi_{n + 1} - \psi^{t}_{n + 1}D_{n}\phi_{n} = D_{0}, & \\
\phi^{t}_{n}D_{n}\psi_{n + 1} - \phi^{t}_{n + 1}D_{n}\psi_{n} = -D_{0}, & \\
\phi^{t}_{n}D_{n}\phi_{n + 1} - \phi^{t}_{n + 1}D_{n}\phi_{n} = 0, &
\end{cases}
\end{equation*}
which can be written in the form 
\begin{equation}
\label{equa.matri.simpl}
\left[
\begin{array}{cc}
\psi^{t}_{n} & \psi^{t}_{n + 1}\\
\phi^{t}_{n} & \phi^{t}_{n + 1}
\end{array}
\right]
\left[
\begin{array}{cc}
D_{n} & 0 \\
0 & D_{n}
\end{array}
\right]
\mathbb{J}
\left[
\begin{array}{cc}
\psi_{n} & \phi_{n} \\
\psi_{n + 1} & \phi_{n + 1}
\end{array}
\right]
=
\mathbb{J}
\left[
\begin{array}{cc}
D_{0} & 0 \\
0 & D_{0}
\end{array}
\right].
\end{equation}

Since $\mathbb{J}^{-1} = -\mathbb{J}$ and 
\[
\left[
\begin{array}{cc}
D_{0} & 0 \\
0 & D_{0}
\end{array}
\right]^{-1} = \left[
\begin{array}{cc}
D^{-1}_{0} & 0 \\
0 & D^{-1}_{0}
\end{array}
\right],
\]
by multiplying to the left both members of identity \eqref{equa.matri.simpl} by  $\mathbb{J}^{-1}\left[
\begin{array}{cc}
D^{-1}_{0} & 0 \\
0 & D^{-1}_{0}
\end{array}
\right]$, it follows that
\[
\left[
\begin{array}{cc}
\psi_{n} & \phi_{n} \\
\psi_{n + 1} & \phi_{n + 1}
\end{array}
\right]
\left[
\begin{array}{cc}
D^{-1}_{0} & 0 \\
0 & D^{-1}_{0}
\end{array}
\right]
\mathbb{J}
\left[
\begin{array}{cc}
\psi^{t}_{n} & \psi^{t}_{n + 1}\\
\phi^{t}_{n} & \phi^{t}_{n + 1}
\end{array}
\right]
=
\mathbb{J}
\left[
\begin{array}{cc}
D^{-1}_{n} & 0 \\
0 & D^{-1}_{n}
\end{array}
\right],
\] where it has been used the fact that
\[
\left[
\begin{array}{cc}
D^{-1}_{0} & 0 \\
0 & D^{-1}_{0}
\end{array}
\right]
\mathbb{J}^{-1}
\left[
\begin{array}{cc}
\psi^{t}_{n} & \psi^{t}_{n + 1}\\
\phi^{t}_{n} & \phi^{t}_{n + 1}
\end{array}
\right]
\left[
\begin{array}{cc}
D_{n} & 0 \\
0 & D_{n}
\end{array}
\right]
\mathbb{J}
\;\;\;\; \mbox{and} \;\;\;\;
\left[
\begin{array}{cc}
\psi_n & \phi_n \\
\psi_{n+1} & \phi_{n+1}
\end{array}
\right]
\]
commute. Such identity can be written in the form
\begin{equation*}
\begin{cases}
\psi_{n}D^{-1}_{0}\phi^{t}_{n} - \phi_{n}D^{-1}_{0}\psi^{t}_{n} = 0,& \\ 
\psi_{n}D^{-1}_{0}\phi^{t}_{n + 1} - \phi_{n} D^{-1}_{0} \psi^{t}_{n + 1}  =  D^{-1}_{n},& \\
\psi_{n + 1}D^{-1}_{0}\phi^{t}_{n} - \phi_{n + 1}D^{-1}_{0}\psi^{t}_{n}  =  - D^{-1}_{n},& \\
\psi_{n + 1}D^{-1}_{0}\phi^{t}_{n + 1} - \phi_{n + 1}D^{-1}_{0}\psi^{t}_{n + 1}  =  0.&
\end{cases}
\end{equation*}
%
%
\end{proof}


\

From these relations, one can obtain the Green Function of the Dirichlet operator $H^{\phi}_{\omega}$.

\begin{proposition}
\label{porp.def.green}
Set, for each $p, q \in \mathbb{Z}_{+}$ and each $z \in \mathbb{C}\setminus\mathbb{R}$,
\[
G_{\omega}^{\phi}(p, q; z) := 
\left\{
\begin{array}{ll}
- \phi_{p}(z, \omega) (D_{0}^\omega)^{-1} F^{t}_{q}(z, \omega), & p \leq q, \\
&\\
- F_{p}(z, \omega) (D_{0}^\omega)^{-1} \phi^{t}_{q}(z, \omega), & p > q,
\end{array}\right.
\]
where $F$ and $\phi$ are, respectively, the Jost solutions at $+\infty$ and the Dirichlet solution to the eigenvalue equation~\eqref{eq.autovalor} at $z$. Then, for each $\textbf{u} \in l^{2}(\mathbb{Z}_{+}, \mathbb{C}^{l})$,
\begin{equation}
\label{eq.resol.green}
\sum_{q} G_{\omega}^{\phi}(p, q; z)\textbf{u}_{q} = ((H_{\omega}^{\phi} - z)^{-1}\textbf{u})_{p}.
\end{equation}
%
\end{proposition}
\begin{proof}
Since for each $n\in\mathbb{Z}_+$, $F_{n}\in l^{2}(\mathbb{Z}_{+}; \mathbb{C}^{l})$, it follows that  
\begin{eqnarray*}
\left(\sum_{q} G_{\omega}^{\phi}(p, q; z)\textbf{u}_{q}\right)_{p \in \mathbb{Z}_{+}}\in l^{2}(\mathbb{Z}_{+}; \mathbb{C}^{l}).
\end{eqnarray*}
One needs to prove that
\[
(H_{\omega}^{\phi} - z)\left(\sum_{q} G_{\omega}^{\phi}(p, q; z)\textbf{u}_{q}\right)_{p \in \mathbb{Z}_{+}} = (\textbf{u}_{p})_{p \in \mathbb{Z}_{+}}.
\]

It follows from the definition of $G_{\omega}^{\phi}(p, q; z)$ that (we omit the dependence of $z$ and $\omega$ in $G_{\omega}^{\phi}(p, q; z)$),
\[
\begin{array}{lll}
(H^{\phi} - z)\left(\sum_{q} G^{\phi}(p, q)\textbf{u}_{q}\right)_{p \in \mathbb{Z}_{+}}= &  & \left(\sum_{q} D_{p}G^{\phi}(p + 1, q)\textbf{u}_{q}\right)_{p \in \mathbb{Z}_{+}}+\\
& & \\ 
& & \left(\sum_{q} D_{p - 1}G^{\phi}(p - 1, q)\textbf{u}_{q}\right)_{p \in \mathbb{Z}_{+}}+  \\
& & \\
& & \left(\sum_{q} (V_{p} - z)G^{\phi}(p, q)\textbf{u}_{q}\right)_{p \in \mathbb{Z}_{+}}.
\end{array}
\]

For a fixed $p\in\mathbb{Z}_+$, one needs to consider the following cases.

\textbf{Case} $q < p - 1$:
\[
\begin{array}{lll}
& D_{p}G^{\phi}(p + 1, q)\textbf{u}_{q} + D_{p - 1}G^{\phi}(p - 1, q)\textbf{u}_{q} + (V_{p} - z)G^{\phi}(p, q)\textbf{u}_{q} = & \\
& & \\
& - D_{p}F_{p + 1}D^{-1}_{0} \phi^{t}_{q}\textbf{u}_{q} - D_{p - 1}F_{p - 1} D^{-1}_{0} \phi^{t}_{q}\textbf{u}_{q} - (V_{p} - z) F_{p} D^{-1}_{0}\phi^{t}_{q}\textbf{u}_{q} = & \\
& & \\
 & - \left[ D_{p}F_{p + 1} + D_{p - 1}F_{p - 1} + (V_{p} - z) F_{p} \right]D^{-1}_{0}\phi^{t}_{q}\textbf{u}_{q}  = - \left[ 0 \right] D^{-1}_{0}\phi^{t}_{q}\textbf{u}_{q} = \textbf{0}. & 
\end{array}
\]

\textbf{Case} $q = p - 1$:
\[
\begin{array}{lll}
& D_{p}G^{\phi}(p + 1, q)\textbf{u}_{q} + D_{p - 1}G^{\phi}(p - 1, q)\textbf{u}_{q} + (V_{p} - z)G^{\phi}(p, q)\textbf{u}_{q} = & \\
& & \\
& - D_{p}F_{p + 1} D^{-1}_{0} \phi^{t}_{q}\textbf{u}_{q} - D_{p - 1}\phi_{p - 1} D^{-1}_{0} F^{t}_{q}\textbf{u}_{q} - (V_{p} - z) F_{p} D^{-1}_{0} \phi^{t}_{q}\textbf{u}_{q} = &  \\
&&\\
& - D_{p}F_{p + 1} D^{-1}_{0} \phi^{t}_{q}\textbf{u}_{q} - D_{p - 1}F_{p - 1}D^{-1}_{0} \phi^{t}_{q}\textbf{u}_{q} - (V_{p} - z) F_{p}D^{-1}_{0}\phi^{t}_{q}\textbf{u}_{q} = &\\
& & \\
& - \left[ D_{p}F_{p + 1} + D_{p - 1}F_{p - 1} + (V_{p} - z) F_{p} \right] D^{-1}_{0}\phi^{t}_{q}\textbf{u}_{q} = - \left[ 0 \right] D^{-1}_{0} \psi^{t}_{q} \textbf{u}_{q} =  \textbf{0},& 
\end{array}
\]
where one has applied \eqref{eq.F.phi}-$(a)$ to the second identity.

\textbf{Case} $q=p$: applying \eqref{eq.F.phi}-$(c)$, one has 
\[
D_{q}G^{\phi}(q + 1, q)\textbf{u}_{q} + D_{q - 1}G^{\phi}(q - 1, q)\textbf{u}_{q} + (V_{q} - z)G^{\phi}(q, q)\textbf{u}_{q} = \textbf{u}_{q}.
\]

\textbf{Case} $q \geq p + 1$:
\[
D_{p}G^{\phi}(p + 1, q)\textbf{u}_{q} + D_{p - 1}G^{\phi}(p - 1, q)\textbf{u}_{q} + (V_{p} - z)G^{\phi}(p, q)\textbf{u}_{q} = \textbf{0}.
\]
\end{proof}


\subsection{Spectral Supports}

We note that, for each $\omega\in\Omega$, the Green Function $G_{\omega}^{\phi}(1, 1; \cdot):\mathbb{C}_{+}\rightarrow M(l,\mathbb{C})$ is a matrix-valued Herglotz function (that is, $G_{\omega}^{\phi}(1, 1; \cdot)$ is analytic and $\Im G_{\omega}^{\phi}(1, 1;z)>0$, for each $z\in\mathbb{C}_+$, given that $G$ is the integral kernel of $(H^\phi_\omega-z)^{-1}$ and $\Im (H^\phi_\omega-z)^{-1}>0)$, from which follows that, for $\kappa$-a.e.~$x\in\mathbb{R}$, 
\[
\lim_{y \downarrow 0}\Im G^{\phi}_\omega(1, 1; x \pm iy)<\infty
\] 
(see~\cite{gesztesy97}). By the Spectral Theorem one has, for each $\textbf{u} \in l^{2}(\mathbb{Z}_{+}; \mathbb{C}^{l})$, 
\begin{equation}
\label{eq.med.espec}
\left\langle  (H^{\phi}_\omega - z)^{-1} \textbf{u}, \textbf{u} \right\rangle = \int \frac{1}{x - z} d\mu_{\textbf{u}}(x),
\end{equation}
where $\mu_{\textbf{u}}$ is a finite Borel measure.

Consider the $l$ canonical vectors $(\textbf{e}_{1,k})_{k = 1,\ldots,l}$ in $(\mathbb{C}^{l})^{\mathbb{Z}_{+}}$, where $(\textbf{e}_{1,k})_{n,j} = \delta_{1,n}\delta_{j,k}$. These vectors form  a spectral basis for the operator $H^{\phi}_{\omega}$, in the sense that
\[\overline{\mathrm{span}\{\cup_{k=1}^l\{(H^\phi_\omega)^n(\textbf{e}_{1,k})\mid n\in\mathbb{Z}_+\}\}}=l^{2}(\mathbb{Z}_{+}; \mathbb{C}^{l});\]


 
\noindent 
therefore, in order to obtain the spectral properties of the operator $H_{\omega}^{\phi}$, it is enough to study the properties of the matrix-valued spectral measure $(\mu_{\textbf{e}_{1, i}, \textbf{e}_{1, j}})_{1 \leq i,j \leq l}$.

The next step consists in obtaining a characterization of the absolutely continuous spectrum (including multiplicity) by
establishing minimal supports for the spectral measures. One says that a Borel subset $S \subseteq \mathbb{R}$ is a minimal support of a positive and finite Borel measure $\mu$ if:
\begin{equation*}
\begin{array}{lll}
(i) & \mu (\mathbb{R} \setminus S) = 0; \\
&\\
  (ii) & \mbox{for each}~S_{0}\subset S~ \mbox{such that}~\mu(S_{0}) = 0,~\kappa(S_{0}) = 0
\end{array}
\end{equation*} 
\noindent where $\kappa(\cdot)$ stands for the Lebesgue measure on $\mathbb{R}$.

Let, for each $z \in \mathbb{C}\setminus \supp[\Omega]$,
\[Q(z)=\int_{\mathbb{R}}\frac{1}{x - z}d\Omega(x)\]
be the matrix-valued Herglotz function associated with the finite matrix-valued Borel measure $\Omega$. It follows from Theorem $6.1$ in \cite{gesztesy97} that 
the sets 
\begin{equation}
\label{def.suport.herglo}
\begin{array}{lll}
S_{ac, r} & := & \{x \in \mathbb{R}\mid \lim_{y \downarrow 0} Q(x + iy)<\infty, \rk[\lim_{y \downarrow 0} \Im[Q(x + iy)]] = r\},\\
&&\\
S_{ac} & := & \bigcup^{l}_{r = 1}S_{ac, r},\\
&&\\
S_{s} & := & \{x \in \mathbb{R}\mid \lim_{y \downarrow 0} \Im[\tr[Q(x + iy)]] = \infty \},
\end{array}
\end{equation}
are minimal supports for the absolutely continuous component of multiplicity $r$, absolutely continuous and singular components of the measure $\Omega$, respectively. 

\begin{proposition}
\label{porp.sup.ac}
For each $z \in \mathbb{C} \setminus \mathbb{R}$ and each $\omega \in \Omega$, let $M^{\phi}(z, \omega)$ be the Weyl-Titchmarsh matrix of the operator $H^\phi_\omega$, defined by \eqref{def.m.weyl}. Then, for each $j=1,\ldots,l$, the set
\[
\Sigma^{\phi,\omega}_{ac, j}  :=  \{x \in \mathbb{R}\mid\lim_{y \downarrow 0} M^{\phi}(x + iy, \omega)<\infty, \rk[\lim_{y \downarrow 0} \Im[M^{\phi}(x + iy, \omega)]] = j\}
\]
is a minimal support for the absolutely continuous component of multiplicity $j$, and the set
\[
\Sigma^{\phi,\omega}_{ac} := \bigcup^{l}_{j = 1} \Sigma^{\phi}_{ac, j}
\]
is a minimal support for the absolutely continuous component of the spectral measure of the operator $H^{\phi}_{\omega}$. Finally, the set
\[
\Sigma^{\phi,\omega}_{s} := \{x \in \mathbb{R}\mid \lim_{y \downarrow 0} \Im[\tr[M^{\phi}(x + iy, \omega)]] = \infty \}
\]
is a minimal support for the singular component of the spectral measure of the operator $H^{\phi}_{\omega}$.
\end{proposition}
\begin{proof}
By definition, 
\begin{equation}
\label{eq.m.herglotz}
G_{\omega}^{\phi}(1, 1; z) = -\phi_{1} D^{-1}_{0}F^{t}_{1} = D^{-1}_{0}D^{t}_{0}(M^{\phi})^{t} = M^{\phi}(z, \omega).
\end{equation}

It follows from identities \eqref{eq.resol.green} and~\eqref{eq.med.espec} that
\[G_{\omega}^{\phi}(1, 1; z)=\int\dfrac{1}{x-z}d\Omega,\]
where $\Omega=(\mu_{\textbf{e}_{1, i}, \textbf{e}_{1, j}})_{1 \leq i,j \leq l}$; thus, since $G_{\omega}^{\phi}(1, 1; z)$ is a matrix-valued Herglotz function, the result follows from~
\eqref{def.suport.herglo} and~\eqref{eq.m.herglotz}. 
\end{proof}

Finally, the next result relates the absolutely continuous spectral components (of all multiplicities) of operators which differ by a finite rank operator.

\begin{proposition}
\label{prop.supor.pertur.z}
Let $(H_{\omega})_\omega$ be the Jacobi operator defined in $l^{2}(\mathbb{Z}; \mathbb{C}^{l})$ by the law~\eqref{eq.ope.din.jacobi}, and let $(H^{\phi}_{\omega, \pm})_\omega$ be the correspondent Dirichlet operators defined in $l^{2}(\mathbb{Z}_{\pm}; \mathbb{C}^{l})$. 
Then, for each $\omega\in\Omega$ and each $j\in\{1,\ldots,2l\}$,  $\sigma_{ac, j}\left( H_{\omega} \right) = \sigma_{ac, j}\left( H_{\omega, +}^{\phi} \oplus H_{\omega, -}^{\phi} \right)$.
\end{proposition}
\begin{proof}
Let $\{ \textbf{e}_{1, 1}, \textbf{e}_{1, 2}, \ldots, \textbf{e}_{1, l}, \textbf{e}_{-1, 1}, \textbf{e}_{-1, 2}, \ldots, \textbf{e}_{-1, l} \}$ be the canonical spectral basis for both $H$ and $\left( H_{+}^{\phi} \oplus H_{-}^{\phi} \right)$. For each $k\in\{1,\ldots,l\}$ and each $m\in\{-1, 1\}$, let $H_{m, k}$ and $\left( H_{+}^{\phi} \oplus H_{-}^{\phi} \right)_{m, k}$ be the operators given by the restrictions of $H$ and $\left( H_{+}^{\phi} \oplus H_{-}^{\phi} \right)$ to the subspace spanned by $\textbf{e}_{m, k}$, respectively.

Kato-Rosenblum's Theorem (see~\cite{simon79}) establishes that if the difference between two bounded self-adjoint operators, say $S$ and $T$, is a finite rank operator, then $\sigma_{ac}(S)=\sigma_{ac}(T)$. 

Since $\left( H_{+}^{\phi} \oplus H_{-}^{\phi} \right)_{m, k}$ is a finite rank perturbation of $H_{m, k}$, it follows that $\sigma_{ac}(H_{m, k})=\sigma_{ac}\left(\left( H_{+}^{\phi} \oplus H_{-}^{\phi} \right)_{m, k}\right)$.
\end{proof}

\section{Lyapunov Exponents and Thouless Formula}\label{thouless}
\zerarcounters

As discussed in Introduction, if the cocycle $(T, A_{z})$, with $A_{z}: \Omega \rightarrow SL(l, \mathbb{C})$ given by \eqref{eq.cociclo.az}, is such that $\log^{+} \left\| A_z(\omega)\right\|: \Omega \rightarrow \mathbb{R} \in L^{1}(\nu)$, then it follows from Oseledec Theorem that its Lyapunov exponents are well defined. 

One may characterize such exponents in terms of the singular values of the transfer matrices $A_n(z,\omega)$ by relation \eqref{eq.ruelle}. We emphazise again that, by the ergocity of $T$, they do not depend on $\omega\in\Omega$. 

If one also assumes that the mapping $D: \Omega \rightarrow M(l, \mathbb{R})$ is bounded, the operator $(H^\phi_\omega)_\omega$ is (as discussed in Section~\ref{resolvente}) in the limit point case at $+\infty$. Hence, one can define the sequence of matrices $(F^{(+)}_{n}(z, \omega))_{n \in \mathbb{Z}_{+}} \in l^{2}(\mathbb{Z}_{+}; M(l, \mathbb{\mathbb{C}}))$ (the Jost solutions), which span the subspace $\mathcal{J}_{+}(z, \omega)$  and satisfies 
\begin{equation}
\label{eq.tran.fmais} 
\left[
\begin{array}{c}
F^{(+)}_{n + 1}(z, \omega)  \\
 \\
D(T^{n}\omega)F^{(+)}_{n}(z, \omega) 
\end{array}
\right]
=
A_{n}(z, \omega)
\left[
\begin{array}{c}
M_{+}(z, \omega)  \\
 \\
D(\omega)
\end{array}
\right].
\end{equation}
One can also define, in the same manner, the sequence $(F^{(-)}_{- n}(z, \omega))_{n \in \mathbb{Z}_{+}} \in l^{2}(\mathbb{Z}_{-}; M(l, \mathbb{\mathbb{C}}))$, which spans the subspace $\mathcal{J}_{-}(z, \omega)$ and satisfies
\begin{equation}
\label{eq.tran.fmenos}
\left[
\begin{array}{c}
F^{(-)}_{- n - 1}(z, \omega)  \\
 \\
D(T^{n}\omega)F^{(-)}_{- n}(z, \omega) 
\end{array}
\right]
=
A_{- n}(z, \omega)
\left[
\begin{array}{c}
M_{-}(z, \omega)  \\
 \\
D(\omega)
\end{array}
\right];
\end{equation}
here, the matrices $A_{- n}(z, \omega)$ are defined as in \eqref{def.mat.trans.erg}.

Note that these subspaces satisfy $\dim(\mathcal{J}_{+}(z, \omega)) = \dim(\mathcal{J}_{-}(z, \omega)) = l$ and their union spans the space of the solutions to the eigenvalue equation~\eqref{eq.autovalor}.  

Since, for each $n\in\mathbb{Z}$ and each $z\in\mathbb{C}$, $A_n(z)\in\mathrm{sp}(2l)$ (that is, $(A_{n}(z))^t\mathbb{J}A_n(z)=\mathbb{J}$), it follows from relation \eqref{eq.ruelle} that the Lyapunov exponents are pairs of symmetric numbers. Moreover, since the sequences $(F^{(\pm)}_{n})_{n \in \mathbb{Z}_{+}}$ are related with the transfer matrices by equations \eqref{eq.tran.fmais} and \eqref{eq.tran.fmenos}, one has (by assuming that the matrices $D(\omega)$ are uniformly bounded), for each $j\in\{1, 2,\ldots, l\}$ and each $z\in\mathbb{C}\setminus\mathbb{R}$,
\[
\lim_{n \rightarrow \infty} \frac{1}{n} \log \left(s_{j}[A_{n}(z, \omega)]\right) 
=
\lim_{n \rightarrow \infty} \frac{1}{n} \log \left(s_{j}[F^{(-)}_{n}(z, \omega)]\right),
\]
and for each $j\in\{l+1,l+2,\ldots, 2l\}$, 
\[
\lim_{n \rightarrow \infty} \frac{1}{n} \log \left(s_{j}[A_{n}(z, \omega)]\right) 
=
\lim_{n \rightarrow \infty} \frac{1}{n} \log \left(s_{j}[F^{(+)}_{n}(z, \omega)]\right).
\]

Let, for each $j\in\{1, 2,\ldots, l\}$ and each $z\in\mathbb{C}\setminus\mathbb{R}$, $\gamma_j^{\pm}(z)$ be the Lyapunov exponents of $F^{(\mp)}_{n}(z, \omega)$. Gathering these results for the cocycle $(T, A_{z})$, one establishes the following proposition.

\begin{proposition}
\label{prop.def.exp.lyap}
Let $(H_{\omega})_{\omega}$ be as above and let $z \in \mathbb{C}\setminus\mathbb{R}$. Then, the cocycle $(T, A_{z})$ has (not necessarily distict) $2l$ Lyapunov exponents which satisfy
\begin{eqnarray*}
 \left\{ \begin{array}{ll} \gamma_j(z)=\gamma^+_j(z), & j=1,2,\ldots,l \\
   \gamma_j(z)=\gamma^-_{2l-j+1}(z), & j=l+1,l+2,\ldots,2l. \end{array}\right.
\end{eqnarray*}
Moreover, for each $j\in\{1,\ldots,l\}$ and for $\nu$-a.e. $\omega\in\Omega$, one has\footnote{For the precise definition of $\Lambda^{j}$, $j\in\{1,\ldots,l\}$, see page $179$ in~\cite{carmona90} and page $43$ in~\cite{petz2014}.} 
\begin{equation*}
\gamma^{+}_{1}(z) + \gamma^{+}_{2}(z) + \ldots + \gamma^{+}_{j}(z) = 
\lim_{n \rightarrow \infty} \frac{1}{n} \log \left\|\Lambda^{j}(F^{(-)}_{n}(z, \omega))\right\|,
\end{equation*}
\begin{equation*}
\gamma^{-}_{1}(z) + \gamma^{-}_{2}(z) + \ldots + \gamma^{-}_{j}(z) = 
\lim_{n \rightarrow \infty} \frac{1}{n} \log \left\|\Lambda^{j}(F^{(+)}_{n}(z, \omega))\right\|,
\end{equation*}
and if $z\in\mathbb{C}$, then for each $j=1,\ldots,l$, $\gamma_j(z)=-\gamma_{2l-j+1}(z)\ge 0$.
\end{proposition}



Set, for each $z\in\mathbb{C}$, 
\begin{equation}\label{gamma}
  \gamma(z):=\sum_{j=1}^l\gamma_j(z);
\end{equation}
it follows from Proposition~\ref{prop.def.exp.lyap} that for $z\in\mathbb{C}\setminus\mathbb{R}$,
\[\gamma(z)=\sum_{j=1}^l\gamma_j^+(z)=-\sum_{j=1}^l\gamma_j^-(z).
\]
One also concludes from Proposition~\ref{prop.def.exp.lyap} that for each $z\in\mathbb{C}\setminus\mathbb{R}$,
$$\gamma^{+}_{1}(z) \geq \ldots \geq \gamma^{+}_{l - 1}(z) \geq \gamma^{+}_{l}(z) \geq 0 \geq  \gamma^{-}_{l}(z) \geq  \gamma^{-}_{l - 1}(z) \geq \ldots \geq \gamma^{-}_{1}(z), 
$$
where, for each $j\in\{1,\ldots,l\}$, $\gamma^{+}_{j}(z) =  - \gamma^{-}_{j}(z)$. 

Our next step consists in establishing Thouless Formula in our setting (Theorem \ref{teo.thou}). 
In order to prove this result, we adapt the proofs presented in \cite{cycon87} and \cite{kotani88}; thus, one needs to define the so-called integrated density of states.

\begin{definition}[Integrated Density of States] 
\label{def.dens.est}
Let $(H_{\omega})_{\omega}$ be as before. One defines the integrated density of states related with this family of operators by the law 
\begin{equation*}
k(\Gamma) :=  \int_{\Omega}\left( \sum_{j = 1}^{l}  \left\langle \mathcal{P}_{\omega}(\Gamma)\textbf{e}_{1,j}, \textbf{e}_{1,j}\right\rangle    \right)d\nu(\omega),
\end{equation*}
where $\Gamma \in \mathcal{A}$ ($\mathcal{A}$ stands for the Borel $\sigma$-algebra of $\mathbb{R}$) and $\mathcal{P}_{\omega}: \mathcal{A} \rightarrow B(l^{2}(\mathbb{Z}; \mathbb{C}^{l}))$ is the resolution of the identity associated with the Dirichlet operator $H^{\phi}_{\omega}$.
\end{definition}

Before we proceed let, for each $N \in \mathbb{N}$ and each $\omega \in \Omega$, $H_{\omega}^{N}$ be the restriction of $H_{\omega}^{\phi}$ to the space $l^{2}(\{0, 1, \ldots, N - 1, N\};\mathbb{C}^l)$. Let also $k_{\omega}^{N}: \mathbb{R} \rightarrow \mathbb{Z}_{+}$ be given by the law
\begin{equation}
\label{def.seq.est}
\begin{array}{lllll}
k_{\omega}^{N}(x) &  := & \dfrac{1}{N} |\{ $eigenvalues of $H_{\omega}^{N}$ least or equal to $x \}|,
\end{array}
\end{equation}
where $|A|$ stands for the cardinality of the set $A$.

In what follows, we show that the sequence of measures $(dk_{\omega}^{N})_{N \in \mathbb{N}}$ converges  weakly to the density of states $dk$.

\begin{proposition}
\label{prop.dens.est}
Let $k$ be as in Definition \ref{def.dens.est}, let $f: \mathbb{R} \rightarrow \mathbb{R}$ be a measurable and bounded function, and let $(k_{\omega}^{N})_{N \in \mathbb{N}}$ be the sequence given by \eqref{def.seq.est}. 
Then, there exists a $\nu$-measurable set $\Omega_{f}$, with $\nu(\Omega_{f}) = 1$, such that, for each $\omega \in \Omega_{f}$,
$$
\lim_{N \rightarrow \infty} \int f(x) dk_{\omega}^{N}(x) = \int f(x) dk(x).
$$
\end{proposition}
\begin{proof}
  Let $N\in\mathbb{N}$; 
  it follows from the definition of trace of an operator that 
\begin{equation}
\label{eq.conv.fra.1}
\int f(x) dk_{\omega}^{N}(x)  
=\frac{1}{N}\dim\img(\chi_N \mathcal{P}_{\omega}(f)\chi_{N})=
\frac{1}{N} \tr[\mathcal{P}_{\omega}(f)\chi_{N}].
\end{equation}
Since $(\textbf{e}_{n,k})$ is an orthonormal basis of $l^{2}(\mathbb{Z}_{+}; \mathbb{C}^{l})$, one has 
\begin{equation}
\label{eq.conv.fra.2}
\frac{1}{N} \tr[\mathcal{P}_{\omega}(f)\chi_{N}] = \frac{1}{N} \sum^{N - 1}_{n = 1}  \sum^{l}_{k = 1} \left\langle \mathcal{P}_{\omega}(f) \textbf{e}_{n, k}, \textbf{e}_{n, k} \right\rangle.
\end{equation}

On the other hand, since the operator $H^{\phi}_{T^{n}\omega}$ is a translation of $H^{\phi}_{\omega}$, it follows that for each $n \in \mathbb{Z}_{+}$,
\begin{equation}
\label{eq.conv.fra.3}
\left\langle \mathcal{P}_{\omega}(f) \textbf{e}_{n, k}, \textbf{e}_{n, k} \right\rangle = \left\langle \mathcal{P}_{T^{n}\omega}(f) \textbf{e}_{1, k}, \textbf{e}_{1, k} \right\rangle.
\end{equation}

So, if one defines, for each $\omega\in\Omega$, 
$$
\overline{f}(\omega) = \sum^{l}_{k = 1} \left\langle \mathcal{P}_{\omega}(f) \textbf{e}_{1, k}, \textbf{e}_{1, k} \right\rangle,
$$
then one gets from relations \eqref{eq.conv.fra.1} to \eqref{eq.conv.fra.3} that
$$
 \int f(x) dk_{\omega}^{N}(x) = \frac{1}{N} \sum^{N - 1}_{n = 1} \overline{f}(T^{n}\omega).
$$

By Birkhoff Ergodic Theorem, there exists a measurable set $\Omega_{f}$ with $\nu(\Omega_{f}) = 1$ such that, for every $\omega \in \Omega_{f}$,
\[
\lim_{N \rightarrow \infty} \frac{1}{N} \sum^{N - 1}_{n = 0} \overline{f}(T^{n}\omega) = \mathbb{E}[\overline{f}].
\]
By the definition of $k$, one has
$$
\mathbb{E}[\overline{f}] = \mathbb{E}\left[\sum^{l}_{k = 1} \left\langle \mathcal{P}_{\omega}(f) \textbf{e}_{1, k}, \textbf{e}_{1, k} \right\rangle \right] = \int f(x)dk(x).
$$

By combining the last three identities, it follows that for every $\omega \in \Omega_{f}$,
$$
\lim_{N \rightarrow \infty} \int f(x) dk_{\omega}^{N}(x) = \int f(x) dk(x).
$$
\end{proof}  

\begin{proposition}
\label{prop.dens.est.2}
Let $(H_{\omega})_{\omega}$, $k$ and $(k_{\omega}^{N})_{N \in \mathbb{N}}$ be as in Proposition \ref{prop.dens.est}. Then, $(dk_{\omega}^{N})_{N \in \mathbb{N}}$ converges weakly to the measure $dk$ for $\nu$-a.e. $\omega\in\Omega$.
\end{proposition}
\begin{proof}
  Let $C_{0}$ denote the space of continuous functions with compact support in $\mathbb{R}$, endowed with the topology of the uniform convergence. Let $(g_{n})_{n}$ be a dense sequence in $C_0$ of real measurable bounded functions. 
  Define  
$$
\Omega_{0} := \bigcap_{n \in \mathbb{N}} \Omega_{g_{n}},
$$
with $\Omega_{g_{n}}$ given by Proposition \ref{prop.dens.est}; then, $\nu(\Omega_{0}) = 1$.

Now, given $f \in C_{0}$ and $\omega \in \Omega_{0}$, it follows from Proposition \ref{prop.dens.est} that
$$
\lim_{N \rightarrow \infty} \int f(x) dk_{\omega}^{N}(x) = \int f(x) dk(x).
$$
\end{proof}

We are finally able to prove Thouless Formula.

\begin{theorem}[Thouless Formula]
\label{teo.thou}
Let $(H_{\omega})_{\omega}$ be as above, let $\gamma(z)$ be given by \eqref{gamma}
and let $k$ be the integrated density of states (Definition \ref{def.dens.est}). Then, for each $z \in \mathbb{C}$,
$$
\gamma(z) = \int_{\mathbb{R}} \log \left|z - x\right| dk(x) - \int_{\Omega} \log \left\vert\det \left( D(\omega) \right)\right\vert d\nu(\omega).
$$
\end{theorem}

\begin{proof}
Firstly, consider the case $z \in \mathbb{C} \setminus \mathbb{R}$. For each $\omega \in \Omega$, let $F^{(\pm)}(z, \omega)$ be the Jost solutions \eqref{def.jost} to the eigenvalue equation~\eqref{eq.autovalor} associated with the operators $H_{\omega,\pm}^\phi$. One has from Proposition \ref{prop.def.exp.lyap} that for $\nu$-a.e. $\omega\in\Omega$, 
$$
\lim_{n \rightarrow \infty} \frac{1}{n} \log \left\vert\det \left( F^{(-)}_{n}(z, \omega)\right)\right\vert = \lim_{n \rightarrow \infty} \frac{1}{n} \log \left\|\Lambda^{j}(F^{(-)}_{n}(z, \omega))\right\|=\gamma(z).
$$

Given that $z \in \rho(H_{\omega})$, the space of solutions to the eigenvalue equation is spanned by the Jost solutions (being, therefore, $2l$-dimensional). If one considers the Dirichlet solution $(\phi_{n}(\omega, z))_n$, each sequence of columns of $(\phi_{n}(\omega, z))_n$ can be written as a linear combination of the sequences of the columns of $(F_{n}^{(+)}(z, \omega))_n$ and $(F_{n}^{(-)}(z, \omega))_n$. More precisely, there exist matrices $B, C$ such that, for each $n \in \mathbb{Z}_{+}$,
\begin{equation}
\label{eq.phi.invert}
\phi_{n}(z, \omega) = F^{(-)}_{n}(z, \omega)B + F^{(+)}_{n}(z, \omega)C.
\end{equation}

Suppose that $B$ is not invertible. Then, there exists $\textbf{v} \neq \textbf{0}$ such that $B\textbf{v} = \textbf{0}$, and consequently, $\phi(z, \omega)\textbf{v} = F^{(+)}(z, \omega)C\textbf{v}\in l^2(\mathbb{N};\mathbb{C}^l)$; hence, $z$ is an eigenvalue of $H^{\phi}_{\omega}$, which is an absurd. This proves that $B$ is invertible. On the other hand, since, for each $z \in \mathbb{C} \setminus \mathbb{R}$ and each $n\in\mathbb{Z}_+$, $F^{(-)}_{n}(z, \omega)$ is invertible, one can rewrite \eqref{eq.phi.invert} as
$$
\phi_{n}(z, \omega) = \left[\mathbb{I} + F_{n}^{(+)}(z, \omega)CB^{-1}\left(F_{n}^{(-)}(z, \omega)\right)^{-1}\right]F^{(-)}_{n}(z, \omega)B.
$$

It follows now from Proposition~\ref{prop.def.exp.lyap} that
\begin{equation}
\label{eq.exp.phi}
\lim_{n \rightarrow \infty} \frac{1}{n} \log \left\vert\det \left( \phi_{n}(z, \omega) \right)\right\vert  = \lim_{n \rightarrow \infty} \frac{1}{n} \log \left\vert\det \left(F^{(-)}_{n}(z, \omega)\right)\right\vert = \gamma(z).
\end{equation}

The Dirichlet solution $\phi$ satisfies, for each $n\in\mathbb{Z}_+$, the recurrence relation
$$
D_{n}\phi_{n + 1} = (z - V_{n})\phi_{n} - D_{n - 1}\phi_{n - 1}.
$$
The entries of the matrices $D_{n}\phi_{n + 1}$ are monic polynomials in $z$ of degree at most $n$. Then, $\det \left(D_{n}\phi_{n + 1}\right)$ is a polynomial of degree $nl$ (see \cite{cycon87}, page $185$, for a proof of this argument).

Note that if $\textbf{v} \neq \textbf{0}$ is such that $\phi_{N + 1}(z, \omega)\textbf{v} = \textbf{0}$, then there exists a solution to the eigenvalue equation in $z$ which is a linear combination of the columns of $(D_{n}\phi_{n}(z, \omega))_{n }$ that vanishes at the entries $0$ and $N + 1$. In this case, $z$ is an eigenvalue of $H^{N}_{\omega}$. Hence, if $D_{N}$ is invertible, then
$\det \left(D_{N}\phi_{N + 1}(z, \omega) \right) = 0$ iff $z$ is an eigenvalue $H^{N}_{\omega}$.

One concludes that $\det \left(D_{N}\phi_{N + 1}(z, \omega) \right)$ is a polynomial in $z$ whose roots are eigenvalues of $H^{N}_{\omega}$; if one denotes them by $(\lambda_{k})_{1 \leq k \leq N l}$, one has
\begin{equation}
\label{thou.1}
\det \left(D_{N}\phi_{N + 1}(z, \omega) \right) = \prod^{Nl}_{k = 1}(z - \lambda_{k}).
\end{equation}

Let, for each $N \in \mathbb{N}$ and each $\omega \in \Omega$, $\beta_{\omega}^{N}$ be the measure defined in the Borel $\sigma$-algebra of $\mathbb{R}$ by the law
\begin{equation*}
\begin{array}{llll}
\beta_{\omega}^{N}(\Lambda)  & = & |\{x \in \Lambda\mid x$ is an eigenvalue of $H^{N}_{\omega}\}|; 
\end{array}
\end{equation*}
then, it follows from \eqref{thou.1} that
$$
\log \left\vert\det \left( D_{N}(\omega) \right)\right\vert +
\log \left\vert\det \left( \phi_{N + 1}(z, \omega) \right)\right\vert = \int \log\left| z - x \right| d\beta_{\omega}^{N}(x).
$$
Consequently, by the definition of $k_{\omega}^{N}$ and by Proposition \ref{prop.dens.est.2} one has, for $\nu$-a.e. $\omega\in\Omega$,
\begin{eqnarray}\label{eq.dnphi}
\nonumber\lim_{N \rightarrow \infty} \frac{1}{N} \left(\log \left\vert\det \left( \phi_{N + 1}(z, \omega) \right)\right\vert+\log \left\vert\det \left( D_N(\omega) \right)\right\vert\right)  &=& 
\lim_{N \rightarrow \infty} \int \log\left| z - x \right| dk_{\omega}^{N}(x)\\
&=&
\int \log\left| z - x \right| dk(x).
\end{eqnarray}

Since $D_{n}(\omega) = D(T^{n}\omega)$ for every $\omega\in\Omega$, it follows from Birkhoff Ergodic Theorem that for $\nu$-a.e. $\omega\in\Omega$,
\begin{equation}
\label{eq.erg.dn}
\lim_{N \rightarrow \infty} \frac{1}{n} \left( \log \left\vert\det \left( D_{N}(\omega) \right)\right\vert \right) = \int_{\Omega} \log \left\vert\det \left( D(\omega)\right)\right\vert d\nu(\omega).
\end{equation}

By combining the identities  \eqref{eq.exp.phi}, \eqref{eq.dnphi} and \eqref{eq.erg.dn},
one gets
\begin{equation}
\label{eq.thou.z.jac.1}
\gamma(z) = \int_{\mathbb{R}} \log \left|z - x\right| dk(x)-\int_{\Omega} \log \left\vert\det \left( D(\omega) \right)\right\vert d\nu(\omega).
\end{equation}

In order to prove Thouless formula for $x \in \mathbb{R}$,  we proceed as in the proof of Theorem~9.20 in~\cite{cycon87}; namely, given $r > 0$ and $x\in \mathbb{R}$, if one takes the mean value of $\gamma(z)$ in the disk $\overline{D}(x;r)$, it follows from \eqref{eq.thou.z.jac.1} that 
\begin{equation}
\label{eq.thou.z.jac.2}
\begin{array}{lll}
& \dfrac{1}{\pi r^{2}} \dint_{\left| x - (p + qi) \right| \leq r} \gamma(p + qi) d\kappa(p)d\kappa(q) = & \\
& & \\
& - \dint_{\Omega} \log \left\vert\det \left( D(\omega)\right)\right\vert d\nu(\omega) & \\
& & \\
& +  \dfrac{1}{\pi r^{2}} \dint_{\left| x - (p + qi)  \right| \leq r} \int_{\mathbb{R}} \log \left|(p + qi) - s\right| dk(s)d\kappa(p)d\kappa(q). &
\end{array}
\end{equation}
Now, it follows from the lemma presented in the proof of Theorem~9.20 in~\cite{cycon87} that the function $f: \mathbb{C} \rightarrow \mathbb{C} \cup \{-\infty\}$, given by the law
$$
f(z) := - \int_{\Omega} \log \left\vert\det \left( D(\omega)\right)\right\vert d\nu(\omega) +
\int \log \left| z - s \right|dk(s),
$$
is subharmonic and so, for every $z_{0} \in \mathbb{C}$, one has
\begin{equation}
\label{eq.prop.sunhar}
f(z_{0}) = \lim_{r \rightarrow 0} \frac{1}{\pi r^{2}} \int_{\left|z - z_{0}\right| \leq r} f(z)dz.
\end{equation}
Thus, by letting $r\downarrow 0$ in both members of \eqref{eq.thou.z.jac.2}, it follows from \eqref{eq.prop.sunhar} that
$$
\gamma(x) = f(x) =  - \int_{\Omega} \log \left\vert\det \left( D(\omega)\right)\right\vert d\nu(\omega) +
\int \log \left| x - s \right|dk(s).
$$
\end{proof}

In what follows, we prove using Thouless Formula that $\gamma(x)$
has, for $\kappa$-a.e. $x\in\mathbb{R}$, a finite ortogonal derivative at $x$. This result is used in the proof of Kotani Theorem~\ref{teo.ko}, as discussed in Section~\ref{kotani}.

\begin{corollary}
\label{coro.thouless}
Let $\gamma(z)$ be as in Theorem \ref{teo.thou}. Then, for $\kappa$-a.e. $x \in \mathbb{R}$, the   normal derivative
$$
\frac{\partial \gamma(x + iy)}{\partial y} = \lim_{y \downarrow 0}\frac{\gamma(x) - \gamma(x + iy)}{y}
$$
exists and is finite.
\end{corollary}
\begin{proof}
It follows from Theorem \ref{teo.thou} that 
\begin{equation*}
\begin{array}{lll}
\lim_{y \downarrow 0} \dfrac{ \gamma(x) - \gamma(x + iy) }{y} & = 
& \lim_{y \downarrow 0} \dint_{\mathbb{R}} \dfrac{ \log \left| x + iy - s \right| - \log \left| x - s  \right| }{ y } dk(s) \\
& & \\
& = & \dint_{\mathbb{R}} i \dfrac{1}{ \left|x - s\right| } dk(s)
\end{array}
\end{equation*} 
Now, since the Borel transform of the measure $k$ exists for $\kappa$-a.e. $x\in\mathbb{R}$, the results follows. 
\end{proof}

\section{Kotani Theorem}\label{kotani}
\zerarcounters

In this section we prove, in our setting, a version of Kotani Theorem, which relates the Lyapunov exponents of the cocycle to the mean value of the singular values of the matrix-valued spectral measure. 
Specifically, we  prove that if some Lyapunov exponent is zero, then the mean value of the correspondent singular value of the matrix-valued spectral measure is strictly positive and finite.
 
Our starting point is the existence of the normal derivative of $\gamma(z)$ (Corollary \ref{coro.thouless}). The idea, again, is to adapt the steps of the proof of Kotani Theorem in \cite{kotani88} (more precisely, we adapt the proof of Theorem~7.2 in~\cite{kotani88}). The main obstacle in our case is the presence of the matrices $D_{n}$ in the recurrence relations. The solution that we present to circunvet this problem is to include these matrices in the term that is dynamically estimated, and then to relate such term to the matrix-valued spectral measure using the Sylvester Law of Inertia (see~\cite{hans66}).

\begin{proposition}[Sylvester Law of Inertia]
\label{prop.sylves}
Let $B$ be an hermitian matrix and let
$$
\begin{array}{lll}
\pi(B) & = & $number of strictly positives eigenvalues of $B; \\ 
\nu(B) & = & $number of strictly negatives eigenvalues of $B; \\
\delta(B) & = & $number of null eigenvalues of $B.
\end{array}
$$
Then, for every non-singular matrix $X$,
$$
\begin{array}{lll}
\pi(B) & = & \pi(X^{*}BX), \\ 
\nu(B) & = &  \nu(X^{*}BX), \\
\delta(B) & = & \delta(X^{*}BX).
\end{array}
$$
\end{proposition}

In the case of dynamically defined operators, one may generalize the Weyl-Titchmarsh function $M^{\phi}$ by the sequence
\begin{equation}
\label{eq.rec.m.din}
M^{\phi}_{n}(z, \omega) = - F_{n  + 1}(z, \omega)  F^{-1}_{n}(z, \omega)  D^{-1}_{n}(\omega),
\end{equation}
where $z \in \mathbb{C}_{+}$, $\omega \in \Omega$, $n\in\mathbb{Z}_+^0$ and $F_{n}(z, \omega) := F^{(+)}_{n}(z, \omega)$  (we omit the index $(+)$ in order to simplify the notation throughout the rest of the text), with $M^{\phi}_{0}(z, \omega)=M^{\phi}(z, \omega)$. It follows from relation~\eqref{eq.rec.m.din} that for each $z \in \mathbb{C}_{+}$, each $\omega \in \Omega$ and each $n\in\mathbb{Z}_+^0$, 
$$
M^{\phi}_{n + 1}(z, \omega) = M^{\phi}_{n}(z, T\omega).
$$

\begin{lemma}
\label{lema.rel.m}
Let $(H_\omega)_{\omega\in\Omega}$ be as before, let $(F_{n}(z, \omega))_{n \in \mathbb{N}}$ be the Jost solutions to the eigenvalue equation~\eqref{eq.autovalor} at $z\in\mathbb{C}$ and let $(M^{\phi}_{n}(z, \omega))_{n \in \mathbb{Z}_{+}^0}$ be the sequence of Weyl-Titchmarsh functions defined by \eqref{eq.rec.m.din}. Then, for each $m,n \in \mathbb{Z}_{+}^0$,
\begin{enumerate}
\item[(a)] $F_{n}(z,T^{m}\omega) F_{m}(z,\omega) = F_{n + m}(z,\omega)$; 
\item[(b)] $M^{\phi}_{0}(z,T^{m}\omega)D_{m}(\omega) = F_{m + 1}(z,\omega)F^{-1}_{m}(z,\omega)$;
\item[(c)] $ - (M^{\phi}_{0})^{-1}(z,T^{n - 1}\omega) - D_{n}(\omega) M^{\phi}_{0}(z,T^{n}\omega) D_{n}(\omega) + V_{n}(\omega) = z \mathbb{I}$. 
\end{enumerate}
\end{lemma}
\begin{proof}
We omit the dependence in $z$ throughout the proof. 

$(a)$ Let $\textbf{a} \in \mathbb{C}^{l}$. The sequence $\textbf{u}=(\textbf{u}_n)_n$, with $\textbf{u}_{n} = F_{n}(\omega)\textbf{a}$, is the solution to the eigenvalue equation~\eqref{eq.autovalor} that satisfies $\textbf{u}_{0} = \textbf{a}$; moreover, $\textbf{u} \in \mathcal{J}_{+}(\omega)$. Then, $\textbf{v}_{n} = F_{n}(T^{m}\omega)\textbf{a}$ corresponds to the solution $\textbf{v} \in \mathcal{J}_{+}(T^{m}\omega)$ to the eigenvalue equation~\eqref{eq.autovalor} that satisfies the initial condition $\textbf{v}_{0} = \textbf{a}$. 

On other hand, $\textbf{w}_{n} = F_{n + m}(\omega)\textbf{a}$ corresponds to the solution $\textbf{u} \in \mathcal{J}_{+}(\omega)$ which satisfies $\textbf{w}_{m} = \textbf{a}$.

One concludes that each column of $F_{n}(T^{m}\omega) F_{m}(\omega)$ coincides with the correspondent column of $F_{n+m}(\omega)$, since both are solutions, in $\mathcal{J}_{+}(T^{m}\omega)$, to the eigenvalue equation~\eqref{eq.autovalor} that satisfy the same initial condition. 

$(b)$ It follows from $(a)$ that for each $m\in\mathbb{Z}_+^0$ and each $\omega\in\Omega$, 
$$
F_{1}(T^{m}\omega) F_{m}(\omega) = F_{m+1}(\omega).
$$
Since $F_{0}(T^{m}\omega) = \mathbb{I}$ and $D_{m}(\omega) = D(T^{m}\omega)$, it follows from relation \eqref{eq.rec.m.din} that 
$$
M^{\phi}_0(T^{m}\omega) = - F_{1}(T^{m}\omega) D^{-1}_{m}(\omega);
$$ 
by combining both relations, one gets 
$$
M^{\phi}_0(T^{m}\omega) = - F_{m + 1}(\omega) F^{-1}_{m}(\omega) D^{-1}_{m}(\omega).
$$

$(c)$ The eigenvalue equation~\eqref{eq.autovalor}, for each $n\in\mathbb{Z}_+$, reads 
$$
D_{n - 1}(\omega)F_{n - 1}(\omega)F^{-1}_{n}(\omega) + D_{n}(\omega)F_{n + 1}(\omega)F^{-1}_{n}(\omega) + V_{n}(\omega) = z \mathbb{I};
$$
it follows from $(b)$ that 
$$
- (M^{\phi}_0)^{-1}(T^{n - 1}\omega) -  D_{n}(\omega) M^{\phi}_0(T^{n}\omega) D_{n}(\omega) + V_{n}(\omega) = z \mathbb{I}.
$$
\end{proof}

Before we present Kotani Theorem in its full generality, we present a partial version of it; namely, we proceed to prove, for $\nu$-a.e. $\omega\in\Omega$, the set inclusion $\overline{\mathcal{Z}}^{ess}_{l} \subseteq \sigma_{ac, l}(H^{\phi}_{\omega})$, with $\mathcal{Z}_{l}$ given by \eqref{def.sup.exp}.

\begin{lemma}
\label{lema.kotani1}
Let, for each $z \in \mathbb{C}_{+}$, $(M_n^{\phi}(z,\omega))_n$ be as in Lemma \ref{lema.rel.m}, and let $\gamma(z)$ be given by \eqref{gamma}. Then, for each $n \in \mathbb{Z}_{+}$,
\begin{equation}
\label{eq.exp.u.y}
\mathbb{E}\left[\log \left\vert\det \left(\mathbb{I} + \frac{\Im[z]}{D_{n}(\omega)\Im[M^{\phi}_{n}(z,\omega)]D_{n}(\omega)}\right)\right\vert\right] = 2\gamma(z).
\end{equation}
\end{lemma}

\begin{proof}
Let $\omega \in \Omega$ and $n \in \mathbb{Z}_+$; if one takes the imaginary part of both members of the identity of Lemma \ref{lema.rel.m}-$(c)$, one gets 
\begin{equation}
\label{recorrencia1}
- D_{n}\Im[M^{\phi}_{n}]D_{n} - \Im[(M^{\phi}_{n - 1})^{-1}] = \Im[z]\mathbb{I}.
\end{equation}
Then,
$$
\mathbb{I} + \frac{\Im[z]}{D_{n}\Im[M^{\phi}_{n}]D_{n}} = ((M^{\phi}_{n - 1})^{*})^{-1}\Im[M^{\phi}_{n - 1}](M^{\phi}_{n - 1})^{-1}(D_{n}\Im[M^{\phi}_{n}]D_{n})^{-1},
$$
from which follows that
$$
\begin{array}{lll}
\mathbb{E}\left[\log \left\vert\det \left(\mathbb{I} + \dfrac{\Im[z]}{D_{n}\Im[M^{\phi}_{n}]D_{n}}\right)\right\vert \right] & = & - \mathbb{E}\left[\log\left\vert \det \left( (M^{\phi}_{n - 1})^{*}\right)\right\vert\right]  + \mathbb{E}\left[\log \det \left( \Im[M^{\phi}_{n - 1}]\right)\right] \\
& & \\
& & - \mathbb{E}\left[\log \left\vert\det\left( M^{\phi}_{n - 1} \right)\right\vert \right] - 2 \mathbb{E}\left[\log \left\vert\det\left( D_{n} \right)\right\vert \right] \\
& & \\
& &  - \mathbb{E}\left[\log \det \left( \Im[M^{\phi}_{n}] \right) \right].
\end{array}
$$
Since
$$
\mathbb{E}\left[\log \det \left( \Im[M^{\phi}_{n - 1}]\right)\right] = \mathbb{E}\left[\log \det \left( \Im[M^{\phi}_{n}] \right) \right]
$$
and, by Proposition~\ref{prop.m},
$$
\mathbb{E}\left[\log \left\vert\det\left( M^{\phi}_{n - 1}\right)\right\vert\right] = \mathbb{E}\left[\log  \left\vert\det\left( (M^{\phi}_{n - 1})^{*} \right\vert\right)\right],
$$
one concludes that  
\begin{equation}
\label{eq.esp.1}
\mathbb{E}\left[\log \det \left\vert\left(\mathbb{I} + \frac{\Im[z]}{D_{n}\Im[M^{\phi}_{n}]D_{n}}\right)\right\vert \right] =  - 2 \mathbb{E}\left[\log \left\vert\det \left( M^{\phi}_{n - 1} \right)\right\vert  +  \log \left\vert\det \left( D_{n} \right)\right\vert \right].
\end{equation}

On other hand, for each $n\in\mathbb{Z}_+$, $M^{\phi}_{n} = - F_{n + 1}F^{-1}_{n}D^{-1}_{n}$, by definition \eqref{eq.rec.m.din}. So, for each $\omega \in \Omega$, one has
$$
\log \left\vert\det \left( M^{\phi}_{n}(\omega) \right)\right\vert= \log\left\vert \det \left( - F_{n + 1} (\omega)\right)\right\vert + \log \left\vert\det \left( F^{-1}_{n}(\omega) \right)\right\vert + \log\left\vert\det \left( D^{-1}_{n}(\omega) \right)\right\vert.
$$

It follows from Birkhoff Ergodic Theorem that for $\nu$-a.e. $\omega \in \Omega$, 
$$
\mathbb{E}\left[\log \left\vert\det \left( M^{\phi}_{n} \right)\right\vert \right] = \lim_{L \rightarrow \infty} \frac{1}{L} \sum^{L - 1}_{n = 0} \log \left\vert\det \left( M^{\phi}_{n}(\omega) \right)\right\vert,
$$
and consequently, that for $\nu$-a.e. $\omega \in \Omega$,
$$
\begin{array}{lll}
\mathbb{E}\left[\log \left\vert\det \left( M^{\phi}_{n} \right)\right\vert \right] & = & \lim_{L \rightarrow \infty} \dfrac{1}{L}\left( \log \left\vert\det \left( F_{L} (\omega)\right)\right\vert - \log \left\vert\det \left( F_{0} (\omega) \right)\right\vert \right) \\
& & \\
& & +  \lim_{L \rightarrow \infty} \dfrac{1}{L} \left( \sum^{L - 1}_{n = 0} \log \left\vert\det \left(  D^{-1}_{n}(\omega) \right)\right\vert \right).
\end{array}
$$

Now,  we compute each term of the right-hand side of this identity separately. Naturally,
$$
\lim_{L \rightarrow \infty} \frac{1}{L} \log \left\vert\det \left( F_{0} (\omega)\right)\right\vert = 0.
$$
By Proposition \ref{prop.def.exp.lyap}, one has 
$$
\lim_{L \rightarrow \infty} \frac{1}{L} \left( \log \left\vert\det \left( F_{L} (\omega)\right)\right\vert\right) = - \gamma(z),
$$
and again by Birkhoff Ergodic Theorem, one has for $\nu$-a.e $\omega\in\Omega$,
$$
\lim_{L \rightarrow \infty} \frac{1}{L}  \sum^{L - 1}_{n = 0} \log\left\vert \det \left(  D^{-1}_{n}(\omega) \right)\right\vert  = - \mathbb{E}\left[ \log\left\vert \det \left( D_{n} \right)\right\vert \right].
$$
Finally,
\begin{equation}
\label{eq.esp.2}
\mathbb{E}\left[\log\left\vert \det \left( M^{\phi}_{n} \right)\right\vert \right] = - \gamma(z) - \mathbb{E}\left[ \log \left\vert\det \left( D_{n} \right)\right\vert \right],
\end{equation}
and the result is now a consequence of relations \eqref{eq.esp.1} and \eqref{eq.esp.2}.
\end{proof}

\begin{proposition}
\label{prop.kotani2}
Let, for each $z \in \mathbb{C}_{+}$, $(M_n^{\phi}(z,\omega))$ and $\gamma(z)$ be as in Lemma~\ref{lema.kotani1}. 
Then, for each $n \in \mathbb{Z}_{+}$,
$$
\mathbb{E}\left[\frac{1}{\tr\left[D_{n}\Im[M^{\phi}_{n}]D_{n} + \frac{\Im[z]}{2}\right]}\right] \leq 2\frac{\gamma(z)}{\Im[z]}.
$$
\end{proposition}
\begin{proof}
Given a square matrix $A$, it is known that\footnote{See, for instance, page $25$ in~\cite{taylor11}.}   
$$
\det \left( e^{A} \right) = e^{\tr[A]}.
$$
Now, if $B$ is such that $\mathbb{I} + B =  e^{A}$\footnote{The matricial logarithm is defined for every matrix with positive eigenvalues; see page $269$ in~\cite{higham08}.}, then  
$$
\det \left( \mathbb{I} + B \right) = e^{\tr\left[\log(\mathbb{I} + B)\right]},
$$ 
that is,  
$$
\log \det \left(  \mathbb{I} + B \right) = \tr\left[ \log(\mathbb{I} + B)\right].
$$
Thus, relation \eqref{eq.exp.u.y} can be written as
$$
\mathbb{E}\left[\tr \left[\log \left(\mathbb{I} + \frac{\Im[z]}{D_{n}\Im[M^{\phi}_{n}]D_{n}}\right)\right]\right] = 2\gamma(z),
$$
which can be expressed in terms of singular the values, $\mu_{k}$, of ${D_{0}\Im[M^{\phi}_{0}]D_{0}}$ as
\begin{equation}\label{eq.singv.le}
\mathbb{E}\left[\sum_{k = 1}^{l} \left(\log \left(1 + \frac{\Im[z]}{\mu_{k}}\right)\right)\right] = 2\gamma(z)
\end{equation}
(namely, given a positive semi-definite matrix $A$ of size $l \times l$, if $f: \mathbb{R}_{+} \rightarrow \mathbb{R}_{+}$ is an increasing function, one has, for each $k\in\{1,\ldots,l\}$, $s_{k}[f(A)] = f(s_{k}[A])$; see Theorem $6.7$ in \cite{petz2014}).

It is also known that for each $x \geq 0$,   
\begin{equation}
\label{des.log}
\log(1 + x) \geq \frac{x}{1 + \frac{x}{2}}.
\end{equation}
Hence,
$$
\mathbb{E}\left[\sum_{k = 1}^{l}  \frac{ \left(\frac{\Im[z]}{\mu_{k}} \right)}{1 + \left(\frac{\Im[z]}{2\mu_{k}}\right)} \right] \leq 2\gamma(z),
$$
and then,
$$
\Im[z] \mathbb{E}\left[\sum_{k = 1}^{l}  \frac{1}{\mu_{k} + \frac{\Im[z]}{2}} \right] \leq 2\gamma(z).
$$
Finally, since the arithmetic mean of a finite set of non-negative numbers is greater than its harmonic mean, it follows that 
$$
\mathbb{E}\left[\frac{1}{\tr \left[D_{n}\Im[M^{\phi}_{n}]D_{n} + \frac{\Im[z]}{2}\right]}\right] \leq
\mathbb{E}\left[\sum_{k = 1}^{l}  \frac{1}{\mu_{k} + \frac{\Im[z]}{2}} \right]
\leq 2\frac{\gamma(z)}{\Im[z]}.
$$
\end{proof}

\


It follows from Corollary \ref{coro.thouless} that the normal derivative $\frac{\partial \gamma(x + iy)}{\partial y}(x)$ exists for $\kappa$-a.e. $x\in\mathbb{R}$. So, for each  $x \in \mathbb{R}$ such that this derivative exists and such that $\gamma(x) = 0$, one has
$$
\frac{\partial \gamma(x + iy)}{\partial y}(x) = \lim_{y \downarrow 0} \frac{\gamma(x + iy) - \gamma(x)}{y} = \lim_{y \downarrow 0} \frac{\gamma(x + iy)}{y}.
$$

It follows from Fatou Lemma and Proposition \ref{prop.kotani2} that, for $\kappa$-a.e. $x\in\mathbb{R}$
$$
\begin{array}{lll}
&\mathbb{E}\left[ \liminf_{y \downarrow 0} \left(\tr\left[D_{n}\Im[M^{\phi}_{n}(x + iy)]D_{n}+ \dfrac{\Im[z]}{2}\right]\right)^{-1}\right] \leq & \\
& & \\
\leq & \limsup_{y \downarrow 0} \mathbb{E}\left[\left(\tr\left[D_{n}\Im[M^{\phi}_{n}(x + iy)]D_{n}+ \dfrac{\Im[z]}{2}\right]\right)^{-1}\right]<\infty.&
\end{array}
$$
Then, for $\kappa$-a.e. $x\in\mathbb{R}$ such that $\gamma(x) = 0$ and for $\nu$-a.e. $\omega \in \Omega$,
$$
\liminf_{y \downarrow 0} \tr\left[ D_{n}\Im[M^{\phi}_{n}(x + iy, \omega)]D_{n} \right] > 0.
$$

On other hand, since for each $n\in\mathbb{Z}_+$, $M^{\phi}(z)$ is a Herglotz function (by \eqref{eq.m.herglotz}), one has for $\kappa$-a.e. $x\in\mathbb{R}$ and each $\omega \in \Omega$, 
$$
\lim_{y \downarrow 0} \tr\left[ D_{n}\Im[M^{\phi}_{n}(x + iy, \omega)]D_{n} \right]<\infty.
$$
Hence, for $\kappa$-a.e. $x\in\mathbb{R}$ such that $\gamma(x) = 0$ and for $\nu$-a.e. $\omega \in \Omega$,
\begin{equation}
\label{rel.sup.mult}
0 < \lim_{y \downarrow 0} \tr \left[D_{n}\Im[M^{\phi}_{n}(x + iy, \omega)]D_{n}\right] < \infty.
\end{equation}

Finally, it follows from 
Proposition \ref{prop.sylves} and relation \eqref{rel.sup.mult} that for $\kappa$-a.e. $x\in\mathbb{R}$ such that $\gamma(x) = 0$ and for $\nu$-a.e. $\omega \in \Omega$, 
$$
0 < \lim_{y \downarrow 0} \tr \left[\Im[M^{\phi}_{n}(x + iy, \omega)] \right] < \infty;
$$ 
the set inclusion $\overline{\mathcal{Z}}^{ess}_{l} \subseteq \sigma_{ac, l}(H^{\phi}_{\omega})$ is now a consequence of Proposition \ref{porp.sup.ac}.


In fact, as in \cite{kotani88}, one can go even further and obtain 
from Proposition \ref{prop.kotani2} a stratified charaterization (in terms of the multiplicity) of the absolutely continuous spectrum with respect to the sets \eqref{def.sup.exp}. 

The idea is to relate the greatest singular values of the Weyl-Titchmarsh matrix-valued function $M^{\phi}$ with the greatest Lyapunov exponents, generalizing, therefore, relation~\eqref{eq.singv.le}.

\begin{proposition}
\label{lema.lyapunov.parcial}
Let $z\in\mathbb{C}_+$, $M^{\phi}(z,\omega)$ and $\gamma(z)$ be as in Lemma \ref{lema.kotani1}. Let $\mu_{1}(\omega) \geq \mu_{2}(\omega) \geq \ldots \geq \mu_{l}(\omega)$ be the singular values of $D_{0}\Im[M^{\phi}(\omega)]D_{0}$. Then, for each $j\in\{1,\ldots,l\}$,
$$
\sum^{j}_{k = 1} \mathbb{E}\left[\log \left(1 + \frac{\Im[z]}{\mu_{k}}\right)\right] \leq 2 (\gamma_{l + 1 - j}(z) + \ldots + \gamma_{l}(z)).
$$
\end{proposition}

\begin{proof}
It follows from \eqref{recorrencia1} that, for each $n \in \mathbb{Z}_+$,
$$
-D_{n}\Im[M^{\phi}_{n}]D_{n} = \Im[z]\mathbb{I} - ((M^{\phi}_{n - 1})^{*})^{-1}\Im[M^{\phi}_{n - 1}](M^{\phi}_{n - 1})^{-1}.
$$
Now, by multiplying both members of this identity to the left by $F^{*}_{n}$ and to the right by $F_{n}$, one gets
$$
- F^{*}_{n}D_{n}\Im[M^{\phi}_{n}]D_{n}F_{n}  =  F^{*}_{n}\Im[z]F_{n} - F^{*}_{n}((M^{\phi}_{n - 1})^{*})^{-1}\Im[M^{\phi}_{n - 1}](M^{\phi}_{n - 1})^{-1}F_{n}.
$$

Since, by definition~\eqref{eq.rec.m.din},
$$
\left\{
\begin{array}{lll}
F^{*}_{n}((M^{\phi}_{n - 1})^{*})^{-1} & =  & - F^{*}_{n - 1}D_{n - 1}, \\
&&\\ 
(M^{\phi}_{n - 1})^{-1}F_{n} & = & - D_{n - 1}F_{n - 1},
\end{array}\right.
$$
the last identity can be written as 
$$
- F^{*}_{n}D_{n}\Im[M^{\phi}_{n}]D_{n}F_{n} = \Im[z]F^{*}_{n}F_{n} -  F^{*}_{n - 1}D_{n - 1}\Im[M^{\phi}_{n - 1}]D_{n - 1}F_{n - 1};
$$
note that this is a relation between two positive semi-definite matrices. By setting $P_{n} := \sqrt{D_{n}\Im[M^{\phi}_{n}]D_{n}}F_{n}$, one has
$$
P^{*}_{n - 1}P_{n - 1} =  \Im[z]F^{*}_{n}F_{n} + P^{*}_{n}P_{n}, 
$$
from which follows the recurrence relation
\begin{equation}
\label{eq.rec.p}
P^{*}_{n - 1}P_{n - 1} = P^{*}_{n}\left(\mathbb{I} + \frac{\Im[z]}{D_{n}\Im[M^{\phi}_{n}]D_{n}}\right)P_{n}.
\end{equation}

Now, we apply the operators 
 $\Lambda^{k}$ to both members of \eqref{eq.rec.p}. For each $j\in\{0,\ldots,l - 1\}$, one gets 
$$
\begin{array}{lll}
\Lambda^{l - j}\left(P^{*}_{n - 1}P_{n - 1}\right) & = & \Lambda^{l - j}\left( P^{*}_{n}\left(\mathbb{I} + \dfrac{\Im[z]}{D_{n}\Im[M^{\phi}_{n}]D_{n}}\right)P_{n} \right)\\
& & \\
& = & \Lambda^{l - j}\left(\left(\mathbb{I} + \dfrac{\Im[z]}{D_{n}\Im[M^{\phi}_{n}]D_{n}}\right) P_{n}P^{*}_{n} \right),
\end{array}
$$
and then,
\begin{equation}
\label{eq.rec.p2}
\left\| \Lambda^{l - j}(P^{*}_{n - 1}P_{n - 1}) \right\| \leq \left\|  \Lambda^{l - j}\left(\mathbb{I} + \frac{\Im[z]}{D_{n}\Im[M^{\phi}_{n}]D_{n}}\right)   \right\|  \left\| \Lambda^{l - j}(P^{*}_{n}P_{n})\right\|.
\end{equation}
Since
$$
\left\| \Lambda^{l - j}\left(\mathbb{I} + \frac{\Im[z]}{D_{n}\Im[M^{\phi}_{n}]D_{n}}\right)   \right\| = \prod^{l - j}_{k = 1}\left(1 + \frac{\Im[z]}{\mu_{l + 1 - k}(T^{n}\omega)}\right),
$$
it follows from \eqref{eq.rec.p2} that
\begin{equation*}
 \log \left\| \Lambda^{l - j}(P^{*}_{n - 1}P_{n - 1}) \right\|  \leq 
 \sum^{l - j}_{k = 1} \log \left(1 + \frac{\Im[z]}{\mu_{l + 1 - k}(T^{n}\omega)}\right)
+  \log \left\| \Lambda^{l - j}(P^{*}_{n}P_{n})\right\|,
\end{equation*}

and so, for each $n\in\mathbb{Z}_+$,
$$
\log \left\|  \Lambda^{l - j}(P^{*}_{0}P_{0}) \right\| \leq  \sum^{n}_{j = 1}\sum^{l - j}_{k = 1}  \log \left(1 + \frac{\Im[z]}{\mu_{l + 1 - k}(T^{j}\omega)}\right) + \log \left\|  \Lambda^{l - j}(P^{*}_{n}P_{n})  \right\|. 
$$

If one multiplies both sides of the last inequality by $\frac{1}{n}$ and let $n \rightarrow \infty$, one gets by Birkhoff Ergodic Theorem that, for $\nu$-a.e. $\omega \in \Omega$, 
\begin{equation}
\label{eq.rec.p3}
0 \leq \sum^{l}_{k = 1} \mathbb{E}\left[ \log \left(1 + \frac{\Im[z]}{\mu_{l + 1 - k}}\right)\right] + 
\lim_{n \rightarrow \infty} \frac{1}{n} \log \left\| \Lambda^{l - j}(P^{*}_{n}P_{n})\right\|.
\end{equation}

Nevertheless, by definition, $F$ and $P$ satisfies for each $n \in \mathbb{Z}_{+}$, 
$$
P^{*}_{n}P_{n} = F^{*}_{n} D_{n}\Im[M^{\phi}_{n}]D_{n} F_{n};
$$
thus, for each $j\in\{0,\ldots,l-1\}$,
$$
\begin{array}{lll}
\Lambda^{l - j} \left( P^{*}_{n}P_{n} \right) & = &  \Lambda^{l - j} \left( F^{*}_{n} D_{n}\Im[M^{\phi}_{n}]D_{n} F_{n} \right) \\
& & \\
& = & \Lambda^{l - j} \left(  D_{n}\Im[M^{\phi}_{n}]D_{n} F_{n}F^{*}_{n} \right),
\end{array}
$$
from which follows that 
\begin{equation*}
\frac{1}{n} \log \left\|  \Lambda^{l - j}(P^{*}_{n}P_{n})  \right\| \leq  \frac{1}{n} \log \left\| \Lambda^{l - j}(D_{n}\Im[M^{\phi}_{n}]D_{n}) \right\| + \frac{1}{n} \log \left\|  \Lambda^{l - j}(F^{*}_{n}F_{n})\right\|.
\end{equation*}

Now, it follows from Proposition \ref{prop.def.exp.lyap} that for $\nu$-a.e. $\omega \in \Omega$, 
$$
- (\gamma_{1}^+(z) + \ldots + \gamma_{l - j}^+(z)) = \lim_{n \rightarrow \infty} \frac{1}{n} \log\left\|\Lambda^{l - j}(F_{n}(\omega, z))\right\|,
$$
and since  $\log \left\| \Lambda^{l - j}(D_{0}\Im[M^{\phi}_{0}]D_{0})(\omega) \right\|<\infty$, one has 
\begin{equation}
\label{eq.rec.p5}
\lim_{n \rightarrow \infty} \frac{1}{n}  \log \left\|  \Lambda^{l - j}(P^{*}_{n}P_{n}) \right\| \leq - 2 (\gamma_{1}(z) + \ldots + \gamma_{l - j}(z)).
\end{equation}
By combining relations \eqref{eq.rec.p3} and \eqref{eq.rec.p5}, one gets 
\begin{equation}
\label{eq.soma.expo.1}
\sum^{l - j}_{k = 1} \mathbb{E}\left[ \log \left(1 + \frac{\Im[z]}{\mu_{l + 1 - k}}\right)\right] \geq 2 (\gamma_{1}(z) + \ldots + \gamma_{l - j}(z)).
\end{equation}

Finally, by Lemma \ref{lema.kotani1} and relation \eqref{eq.soma.expo.1}, it follows that
$$
\sum^{j}_{k = 1} \mathbb{E}\left[ \log \left(1 + \frac{\Im[z]}{\mu_{k}}\right)\right] = \sum^{l}_{k = l - j + 1} \mathbb{E}\left[ \log \left(1 + \frac{\Im[z]}{\mu_{l + 1 - k}}\right)\right] \leq 2 (\gamma_{l + 1 - j}(z) + \ldots + \gamma_{l}(z)).
$$
\end{proof}

Before we present the proof of Kotani Theorem, some final preparation is required.

Suppose now that for a given $x \in \mathbb{R}$, there exists a solution to the eigenvalue equation~\eqref{eq.autovalor} for the operator $H_{\omega}$ which is square-sumable. Given $y > 0$, since $x + iy\in\rho(H_\omega)$, the correspondent solution for $x + iy$ also exists. Then, by Green Formula, one can compare the norms of these solutions in $l^{2}(\mathbb{N}; \mathbb{C}^{l})$.

\begin{lemma}
\label{lema.m.desi}
Let $\omega \in \Omega$, $k\in\{1,\ldots,k\}$, $n \in \mathbb{N}$, and let $\textbf{f}_{n}^{(k)}(x + iy, \omega)$ be the $k$-th column of the matrix $F_{n}(x + iy, \omega)$. If 
$\textbf{f}^{(k)}(x, \omega)\in l^2(\mathbb{N};\mathbb{C}^l)$, then
$$	
\sum^{\infty}_{m = 1} 
\left\| \textbf{f}_{m}^{(k)}(x + iy) \right\|_{\mathbb{C}^{l}}^{2}
\leq
\sum^{\infty}_{m = 1}
\left\| \textbf{f}_{m}^{(k)}(x) \right\|_{\mathbb{C}^{l}}^{2}.
$$
\end{lemma}

\begin{proof}
We omit the dependence in $\omega$ throughout the proof. Let $z = x + iy$. By applying Lemma \ref{lema.wrons.zx} to $\textbf{f}^{(k)}(z)$ and $\textbf{f}^{(k)}(x)$, one gets
$$
\Im[z]\left(2\left\| \textbf{f}^{(k)}(z) \right\|^{2} - \left\langle \textbf{f}^{(k)}(z) , \textbf{f}^{(k)}(x) \right\rangle  - \left\langle \textbf{f}^{(k)}(x) , \textbf{f}^{(k)}(z)\right\rangle \right) = 0;
$$
namely, given that these solutions are square-sumable at $+\infty$, one has
$$
\lim_{n \rightarrow \infty}W_{[\textbf{f}^{(k)}(z) - \textbf{f}^{(k)}(x) , \overline{\textbf{f}}^{(k)}(z) - \overline{\textbf{f}}^{(k)}(x)]}(n + 1) = 0,
$$
and since $\textbf{f}^{(k)}_{0}(z) - \textbf{f}^{(k)}_{0}(x) = \textbf{0}$, one also has
$$
W_{[\textbf{f}^{(k)}(z) - \textbf{f}^{(k)}(x), \overline{\textbf{f}}^{(k)}(z) - \overline{\textbf{f}}^{(k)}(x)]}(1) = 0.
$$
Hence, 
$$
\begin{array}{lllll}
\left\| \textbf{f}^{(k)}(z) - \textbf{f}^{(k)}(x)  \right\|^{2} & = & \left\| \textbf{f}^{(k)}(z) \right\|^{2} + \left\| \textbf{f}^{(k)}(x) \right\|^{2} \\
& & \\
& & - \left\langle \textbf{f}^{(k)}(z) , \textbf{f}^{(k)}(x) \right\rangle  - \left\langle \textbf{f}^{(k)}(x) , \textbf{f}^{(k)}(z)\right\rangle \\
& & \\
& = & \left\| \textbf{f}^{(k)}(x) \right\|^{2} - \left\| \textbf{f}^{(k)}(z) \right\|^{2},
\end{array}
$$
and consequently,
$$
\left\|\textbf{f}^{(k)}(x) \right\|^{2} - \left\| \textbf{f}^{(k)}(z) \right\|^{2} > 0.
$$
\end{proof}

\begin{corollary}
\label{coro.exp.dec}
Let $z\in\mathbb{C}_+$, $(H_{\omega})_{\omega}$ and $M^{\phi}(z,\omega)$ be as in Lemma \ref{lema.kotani1}. Suppose that there exist $x \in \mathbb{R}$ and $j\in\{1,\ldots, l\}$ such that $\gamma_{j}(x) > 0$. Then, for each $y > 0$, $\gamma_{j}(x) \leq \gamma_j(x + iy)$.
\end{corollary}
\begin{proof}
If $\gamma_{j}(x) > 0$ then, for $\nu$-a.e. $\omega\in\Omega$, there exists a solution to eigenvalue equation at $x$, say $\textbf{f}(x, \omega)$, which decays exponentially fast. By Lemma \ref{lema.m.desi}, it follows that for $\nu$-a.e. $\omega$, the correspondent solution $\textbf{f}(x + iy, \omega)$ decays at least with the same exponential rate. Hence, $\gamma_{j}(x) \leq \gamma_{j}(x + iy)$.
\end{proof}

\begin{proposition}
\label{prop.kotani}
Let $(H_{\omega})_{\omega}, M^{\phi}(\omega), \gamma$ and $(\mu_{k})^{l}_{k = 1}$ be as in Proposition \ref{lema.lyapunov.parcial}. Let $x \in \mathbb{R}$, $j\in\{1,\ldots, l\}$, and suppose that there exists, for each $k = 1, 2, \ldots, l - j$, the limit
$$
\gamma_{k}(x) = \lim_{y \downarrow 0} \gamma_{k}(x + iy).
$$ 
Then, for each $y > 0$,
$$
\begin{array}{lll}
\mathbb{E}\left( \dsum^{j}_{k = 1} \dfrac{1}{\mu_{k}(x + iy) + \frac{y}{2}} \right) & \leq & \dfrac{2}{y} \left(\dsum^{j}_{k = 1}\gamma_{l + 1 - k}(x + iy) \right) \\
& & \\
& & + \dfrac{2}{y} \left( \dsum^{l}_{k = j + 1}(\gamma_{l + 1 - k}(x + iy) - \gamma_{l + 1 - k}(x))  \right).
\end{array}
$$
\end{proposition}

\begin{proof}
It follows from Proposition \ref{lema.lyapunov.parcial} that
$$
\sum^{j}_{k = 1} \mathbb{E}\left[\log \left(1 + \frac{y}{\mu_{k}}\right)\right] \leq 2 (\gamma_{l + 1 - j}(x+iy) + \ldots + \gamma_{l}(x+iy)).
$$
Now,  by employing the same arguments presented in the proof of Proposition \ref{prop.kotani2}, one gets 
\begin{equation}
\label{eq.soma.ex1}
\mathbb{E}\left( \sum^{j}_{k = 1} \frac{1}{\mu_{k}(x + iy) + \frac{y}{2}} \right) \leq \frac{2}{y} \sum^{j}_{k = 1}\gamma_{l + 1 - k}(x + iy). 
\end{equation}

On the other hand, it follows from Corollary \ref{coro.exp.dec} that
$$
\sum^{l}_{k = j + 1}( \gamma_{l + 1 - k}(x + iy) - \gamma_{l + 1 - k}(x) ) \geq 0,
$$
and since $y > 0$, one has
\begin{equation}
\label{eq.soma.ex2}
0 \leq \frac{2}{y} \sum^{l}_{k = j + 1}( \gamma_{l + 1 - k}(x + iy) - \gamma_{l + 1 - k}(x) ).
\end{equation}

The result now follows from relations \eqref{eq.soma.ex1} and \eqref{eq.soma.ex2}. 
\end{proof}

\begin{theorem}[Kotani Theorem]
\label{teo.ko}
Let $(H_\omega)_\omega$ be the family of ergodic matrix-valued Jacobi operators of the form~\eqref{eq.ope.din.jacobi} such that the mapping $D: \Omega \rightarrow GL(l, \mathbb{R})$ is bounded, and for each $\omega\in\Omega$, $D(\omega)$ is a symmetric and invertible $l\times l$ matrix. Suppose also that the mapping $A_z$, given by the law~\eqref{eq.cociclo.az}, is such that $\log^{+} \left\| A_z (\omega)\right\|\in L^{1}(\nu)$. Let also, for each $j\in\{1,\ldots,l\}$, $\mathcal{Z}_{j}$ be the set given by \eqref{def.sup.exp}. Then, for $\nu$-a.e. $\omega\in\Omega$, the restriction of the absolutely continuous spectrum of $H_\omega$ to the set $\mathcal{Z}_{j}$ has multiplicity at least $2j$.
\end{theorem}

\begin{proof}
Define the sets
$$
\begin{array}{lll}
Q_{1} & := & \{x \in \mathbb{R}\mid \exists \lim_{y \downarrow 0} \gamma(x + iy)\}, \\
& & \\
Q_{2} & := & \{x \in \mathbb{R}\mid \exists \lim_{y \downarrow 0} \frac{\partial \gamma}{\partial y}(x + iy) \}, \\
& & \\
Q_{3} & := & \{x \in Q_{1} \cap Q_{2}\mid  \lim_{y \downarrow 0}\left[\frac{\gamma(x + iy) - \gamma(x)}{y} - \frac{\partial \gamma}{\partial y}(x + iy) \right]= 0 \}, \\
& & \\
Q_{4} & := & \{x \in \mathbb{R}\mid \lim_{y \downarrow 0} \Im[M(x + iy, \omega)]<\infty, $ for $\nu$-a.e. $\omega\in\Omega \}.
\end{array}
$$

It follows from Corollary \ref{coro.thouless} that for $\nu$-a.e. $\omega \in \Omega$, $\kappa(\mathbb{R} \setminus Q_{3}) = 0$. One has, from the definitions of $\mathcal{Z}_j$ and $Q_3$, that for each $x_{0}\in \mathcal{Z}_{j} \cap Q_{3}$, there exists the limit
$$
\lim_{y \downarrow 0} \frac{1}{y} \left[\sum^{j}_{k = 1}\gamma_{k}(x_{0} + iy) + \sum^{l}_{k = j + 1}(\gamma_{k}(x_{0} + iy) - \gamma_{k}(x_{0}))\right] = \lim_{y \downarrow 0} \frac{1}{y}[\gamma(x_{0} + iy) - \gamma(x_{0})].
$$

Hence, it follows from Proposition \ref{prop.kotani} and Fatou Lemma that, for each $x_{0} \in \mathcal{Z}_{j} \cap Q_{3}$ and each $k\in\{1,\ldots,j\}$,  
$$
\mathbb{E}\left[\liminf_{y \downarrow 0}\frac{1}{\mu_{k}(x_{0} + iy) + \frac{y}{2}}\right] \le \limsup_{y \downarrow 0} \mathbb{E}\left[\frac{1}{\mu_{k}(x_{0} + iy) + \frac{y}{2}}\right] < \infty,
$$
%
and so
$$
\liminf_{y \downarrow 0} \mu_{k}(x_{0} + iy, \omega) > 0.
$$  

On the other hand, given that $M^{\phi}$ is a Herglotz function (by \eqref{eq.m.herglotz}), it follows that $\kappa(\mathbb{R} \setminus Q_{4}) = 0$. Summing up,  if $x_{0} \in \mathcal{Z}_{j} \cap Q_{3} \cap Q_{4}$, then for $\nu$-a.e. $\omega\in\Omega$ and each $k\in\{1,\ldots,j\}$, one has
$$
0 < \lim_{y \downarrow 0} \mu_{k}(x_{0} + iy, \omega) < \infty.
$$

Therefore, one concludes from these results that
$$
\rk \left[ D_{0}(\omega)\Im[M(x_{0}, \omega)]D_{0}(\omega)  \right] \geq j.
$$

Now, it follows from Proposition \ref{prop.sylves} that
\[\rk\Im[M(x_{0}, \omega)]=\rk \left[ D_{0}(\omega)\Im[M(x_{0}, \omega)]D_{0}(\omega)  \right] \geq j.\] 

By analogous arguments, one can prove the same results for the family $(H^{\phi}_{\omega,-})_{\omega}$ defined in $l^2(\mathbb{Z}_-;\mathbb{C}^l)$ and then, by Proposition \ref{prop.supor.pertur.z}, one can extend, with multiplicity $2j$, the result to the family of operators $(H_{\omega})_{\omega}$ defined in $l^2(\mathbb{Z};\mathbb{C}^l)$.
\end{proof}

\section{Ishii-Pastur Theorem}\label{pastur}
\zerarcounters

Now, we discuss the Ishii-Pastur Theorem (proved for scalar Schr\"odinger operators in \cite{ishii73} and \cite{pastur80}), the reciprocal of Kotani Theorem. In the context of matrix-valued Schr\"odinger operators, the result is presented in \cite{kotani88}. 

\begin{proposition}
\label{prop.m.real}
Let $M^{\phi}(z, \omega)$ be the Weyl-Titchmarsh function for the operator $H^{\phi}_{\omega}$ and let $\textbf{f}^{(k)}(z)$ be as in Lemma \ref{lema.m.desi}. Let $\omega \in \Omega$, $x \in \mathbb{R}$, $k\in\{1,\ldots,l\}$, and suppose that $\textbf{f}^{(k)}(x)\in l^2(\mathbb{N};\mathbb{C}^l)$. Then,
$$
\lim_{y \downarrow 0} (D_{0}\Im[M^{\phi}(x + iy, \omega)]D_{0})_{kk} = 0. 
$$ 
\end{proposition}

\begin{proof}
It follows from Proposition \ref{prop.m}-$(b)$ and Lemma \ref{lema.m.desi} that, for each $y > 0$, 
$$
(D_{0}\Im[M(x + iy, \omega)]D_{0})_{kk}  =  y \sum^{\infty}_{n = 1}  (F_{n}^{*}F_{n})_{kk} =  y \sum^{\infty}_{n = 1} \left\| \textbf{f}_{n}^{(k)}(x + iy) \right\|_{\mathbb{C}^{l}}^{2}\leq \sum^{\infty}_{n = 1} \left\| \textbf{f}_{n}^{(k)}(x) \right\|_{\mathbb{C}^{l}}^{2}. 
$$

Therefore, 
$$
\lim_{y \downarrow 0} \frac{(D_{0}\Im[M(x + iy, \omega)]D_{0})_{kk}}{y} \leq \sum^{\infty}_{n = 1} \left\| \textbf{f}^{(k)}_{n}(x) \right\|_{\mathbb{C}^{l}}^{2} < \infty,
$$
from which follows that $\lim_{y \downarrow 0} (D_{0}\Im[M(x + iy, \omega)]D_{0})_{kk} = 0$.
\end{proof}

\begin{theorem}[Ishii-Pastur Theorem]
\label{teo.pastur}
Let $(H_\omega)_\omega$ be the family of ergodic matrix-valued Jacobi operators of the form~\eqref{eq.ope.din.jacobi} such that the mapping $D: \Omega \rightarrow GL(l, \mathbb{R})$ is bounded, and for each $\omega\in\Omega$, $D(\omega)$ is a symmetric and invertible $l\times l$ matrix. Suppose also that the mapping $A_z$, given by the law~\eqref{eq.cociclo.az}, is such that $\log^{+} \left\| A_z (\omega)\right\|\in L^{1}(\nu)$. Let also, for each $j\in\{1,\ldots,l\}$, $\mathcal{Z}_{j}$ be the set given by \eqref{def.sup.exp}. Then, for $\nu$-a.e. $\omega\in\Omega$, the restriction of the absolutely continuous spectrum of $H_\omega$ to the set $\mathcal{Z}_{j}$ has multiplicity at most $2j$.
\end{theorem}
\begin{proof}
Let $x \in \mathcal{Z}_{j}$. It follows from Oseledets Theorem that, for $\nu$-a.e. $\omega\in\Omega$, there exist  $l - j + 1$ linearly independent solutions to the eigenvalue equation~\eqref{eq.autovalor} (namely, $\textbf{f}^{(k)}(x)$) that decay exponentially fast at $+\infty$.

So, it follows from Proposition \ref{prop.m.real} that for $\nu$-a.e. $\omega\in\Omega$ and for $l - j + 1$ values of $k$, 
$$
\lim_{y \downarrow 0} (D_{0}\Im[M^{\phi}(x + iy)]D_{0})_{kk} = 0,
$$ 
and then, by Proposition \ref{prop.sylves}, that
$$
\lim_{y \downarrow 0} (\Im[M^{\phi}(x + iy)])_{kk} = 0.
$$
Thus, for $\nu$-a.e. $\omega\in\Omega$, one has 
$$
\rk[\lim_{y \downarrow 0} \Im[M^{\phi}(x + iy)]] \leq j.
$$

Now, as in the proof of Theorem \ref{teo.ko}, the same conclusion is valid for the family $(H^{\phi}_{\omega,-})_{\omega}$ of operators defined in $l^2(\mathbb{Z}_-;\mathbb{C}^l)$, and again by Proposition \ref{prop.supor.pertur.z}, one can extend the result, with multiplicity $2j$, 
to the family of operators $(H_{\omega})_{\omega}$ defined in $l^2(\mathbb{Z};\mathbb{C})$. 
\end{proof}

\bibliography{bibfile}{}

\begin{thebibliography}{10}

\bibitem{carmona90}
{\sc Carmona, R., and Lacroix, J.}
\newblock {\em Spectral Theory of Random {S}chr{\"o}dinger Operators}.
\newblock Birkh{\"a}user, Boston, 1990.

\bibitem{cycon87}
{\sc Cycon, H., R.C.Froese, Kirsch, W., and B.Simon}.
\newblock {\em {S}chr{\"o}dinger Operators: with Application to Quantum
  Mechanics and Global Geometry}.
\newblock Springer-Verlag, Berlin, 1987.

\bibitem{cesar2009}
{\sc de~Oliveira, C.~R.}
\newblock {\em Intermediate Spectral Theory and Quantum Dynamics}.
\newblock Birkh{\"a}user, Basel, 2009.

\bibitem{gesztesy97}
{\sc Gesztesy, F., and Tsekanovskii, E.}
\newblock On matrix–valued herglotz functions.
\newblock {\em Mathematische Nachrichten 218}, 1 (2000), 61--138.

\bibitem{haro13}
{\sc Haro, A., and Puig, J.}
\newblock A {T}houless formula and {A}ubry duality for long-range
  {S}chr{\"o}dinger skew-products.
\newblock {\em Nonlinearity}, 26 (2013), 1163--1187.

\bibitem{petz2014}
{\sc Hiai, F., and Petz, D.}
\newblock {\em Introduction to Matrix Analysis and Applications}.
\newblock Springer, New Delhi, 2014.

\bibitem{higham08}
{\sc Higham, N.~J.}
\newblock {\em Functions of Matrices, Theory and Computation}.
\newblock Societv for Industrial and Applied Mathematics, Philadelphia, 2008.

\bibitem{ishii73}
{\sc Ishii, K.}
\newblock Localization of eigenstates and transport phenomena in the one
  dimensional disordered system.
\newblock {\em Progress of Theoretical Physics Supplement}, 53 (1973), 77--138.

\bibitem{kotani82}
{\sc Kotani, S.}
\newblock Lyapunov indices determine absolutely continuous spectra of
  stationary random one-dimensional {S}chr{\"o}dinger operators.
\newblock {\em Proc. Kyoto Stoch. Conf.\/} (1982).

\bibitem{kotani88}
{\sc Kotani, S., and Simon, B.}
\newblock Stochastic {S}chr{\"o}dinger operators and {J}acobi matrices on the
  strip.
\newblock {\em Communications in Mathematical Physics}, 119 (1988), 403--429.

\bibitem{kunz1981}
{\sc Kunz, H., and Souillard, B.}
\newblock Sur le spectre des op\'erateurs aux diffrences finies al\'eatoires.
\newblock {\em Communications in Mathematical Physics}, 78 (1981), 201--246.

\bibitem{marx15}
{\sc Marx, C.~A., and Jitomirskaya, S.}
\newblock Dynamics and spectral theory of quasi-periodic {S}chr{\"o}dinger-type
  operators.
\newblock {\em Ergodic Theory and Dynamical Systems}, 37 (2017), 2353--2393.

\bibitem{oseledets68}
{\sc Oseledets, V.~I.}
\newblock A multiplicative ergodic theorem. characteristic {L}japunov,
  exponents of dynamical systems.
\newblock {\em Trudy Moskovskogo Matematicheskogo Obshchestva}, 19 (1968),
  179--210.

\bibitem{pastur80}
{\sc Pastur, L.}
\newblock Spectral properties of disordered systems in the one-body
  approximation.
\newblock {\em Communications in Mathematical Physics}, 75 (1980), 179--196.

\bibitem{simon79}
{\sc Reed, M., and Simon, B.}
\newblock {\em Methods of Modern Mathematical Physics III: Scattering Theory}.
\newblock Academic Press, San Diego, 1979.

\bibitem{ruelle79}
{\sc Ruelle, D.}
\newblock Ergodic theory of differentiable dynamical systems.
\newblock {\em Institut des Hautes \`Etudes Scientifiques}, 50 (1979), 275.

\bibitem{hans66}
{\sc Schneider, H.}
\newblock Topological aspects of {S}ylvester's theorem on the inertia of
  hermitian matrices.
\newblock {\em The American Mathematical Monthly 73}, 8 (1966), 817--821.

\bibitem{simon83}
{\sc Simon, B.}
\newblock Kotani theory for one dimensional stochastic jacobi matrices.
\newblock {\em Communications in Mathematical Physics 89\/} (1983).

\bibitem{taylor11}
{\sc Taylor, M.~E.}
\newblock {\em {P}artial {D}ifferential {E}quations I}.
\newblock Springer, New York, 2011.

\bibitem{teschl00}
{\sc Teschl, G.}
\newblock {\em {J}acobi Operators and Completely Integrable Nonlinear Lattices,
  Mathematical Surveys and Monographs 72}.
\newblock American Mathematical Society, Providence, 2000.

\bibitem{thou1972}
{\sc Thouless, D.}
\newblock A relation between the density of states and range of localization
  for onedimensional random system.
\newblock {\em Journal of Physics}, C 5 (1972), 77--81.

\end{thebibliography}
\bibliographystyle{acm}

\end{document}